\documentclass[journal]{IEEEtran}
\usepackage{amsmath,graphicx,amssymb,amsthm,amsfonts,mathrsfs,cite}
\usepackage{soul}
\usepackage{color}
\newtheorem{definition}{Definition}
\newtheorem{proposition}{Proposition}
\newtheorem{theorem}{Theorem}
\newtheorem{lemma}{Lemma}

\long\def\symbolfootnote[#1]#2{\begingroup%
	\def\thefootnote{\fnsymbol{footnote}}\footnote[#1]{#2}\endgroup} 

\begin{document}
\title{On the Adversarial Robustness of Subspace Learning}

\author{\IEEEauthorblockN{Fuwei Li, Lifeng Lai, and Shuguang Cui}\\
}
\maketitle
\symbolfootnote[0]{
	F. Li, L. Lai and S. Cui are with the Department of Electrical and Computer Engineering, University of California, Davis, CA. Email:\{fli,lflai,sgcui\}@ucdavis.edu.
	The work of F. Li and S. Cui was supported by the National Science Foundation with Grants CNS-1824553, DMS-1622433, AST-1547436, and ECCS-1659025.
	The work of L. Lai was supported by the National Science Foundation under Grants CCF-1717943, CNS-1824553 and CCF-1908258. This paper was presented in part at IEEE International Conference on Acoustics, Speech, and Signal Processing, Brighton, UK, May. 2019\cite{li2019adversarial}.
}
\pagestyle{plain}
\begin{abstract}
In this paper, we study the adversarial robustness of subspace learning problems. Different from the assumptions made in existing work on robust subspace learning where data samples are contaminated by gross sparse outliers or small dense noises, we consider a more powerful adversary who can first observe the data matrix and then intentionally modify the whole data matrix.
We first characterize the optimal rank-one attack strategy that maximizes the subspace distance between the subspace learned from the original data matrix and that learned from the modified data matrix. We then generalize the study to the scenario without the rank constraint and characterize the corresponding optimal attack strategy. Our analysis shows that the optimal strategies depend on the singular values of the original data matrix and the adversary's energy budget.  
Finally, we provide numerical experiments and practical applications to demonstrate the efficiency of the attack strategies. 
\end{abstract}


\section{Introduction}
Subspace learning has a wide range of applications, such as surveillance video analysis, recommendation system, anomaly detection, etc\cite{Li17compressvideo,xin11subspacedoa,shen17onlinebigdata,guo2014online,otazo2015low,koren2009matrix,mardani2012dynamic,guo2016video}. Among a large variety of subspace learning algorithms, principal component analysis (PCA) is one of the most widely used algorithms. In this paper, we will use PCA as the subspace learning algorithm. PCA computes a small number of principal components, which are orthogonal to each other and represent the majority of the variability of the data samples, and treats the span of these principal components as the desired low-dimensional subspace. Furthermore, many works have proposed robust PCA that can mitigate the impact of certain percentages of outliers and small dense random noise\cite{candes2011robust,hsu2011robust,qiu2014recursive,chen2011robust}.



In this paper, we investigate the adversarial robustness of subspace learning algorithms. Particularly, we examine the robustness of subspace learning algorithms against not only random noise or unintentional corrupted data as considered in existing works, but also malicious data produced by powerful adversaries who can modify the whole data set. Our study is motivated by the fact that subspace learning and many other machine learning algorithms are increasingly being used in safety critical and security related applications, such as autonomous vehicle system\cite{eykholt2017robust}, voice recognition\cite{carlini2016hidden}, medical image processing\cite{finlayson2019adversarial}, etc. In these applications, there might exist powerful adversaries who can modify the data with the goal of maneuvering the machine learning algorithms to make the wrong decision or leave backdoor in the system\cite{chen2017targeted}. To ensure the security and safety of these systems, it is important to understand the impact of these adversarial attacks on the performance of machine learning algorithms.   

In our problem, given the original data matrix, we learn a low-dimensional subspace via PCA. However, there is an adversary who can observe the whole data matrix and then carefully design a modification matrix to change the original data. The goal of the adversary is to modify the original data so as to maximize the subspace distance between the subspace learned from the original data and that learned from the modified data. In this paper, we use Asimov distance\cite{golub2012matrix}, defined as the largest principal angle between two subspaces, to measure the subspace distance. Asimov distance has a close relationship with the chordal 2-norm distance and the Finsler distance, which are used in the analysis of optimization on Manifolds~\cite{edelman1998geometry,weinstein2000almost}. Additionally, Asimov distance is closely related to the projection 2-norm distance and the gap distance, which are used in the control theory to describe the stability and robustness of a system~\cite{georgiou1990optimal,vinnicombe1993frequency,qui1992feedback}. As the Asimov distance depends on the modification matrix in a complex manner, to characterize the optimal attack strategy that maximizes the Asimov distance, we need to solve a complicated non-convex optimization problem.    

Towards this goal, we first solve the optimization problem with an additional rank-one constraint on the modification matrix. We note that a rank-one modification is already powerful enough to capture many common modifications such as changing one data sample, inserting one adversarial data, deleting one feature, etc. Furthermore, the techniques and insights obtained from this special case are useful for the general case without the rank-one constraint. For the rank-one attack case, we study two different scenarios depending on whether the dimension of the selected subspace is equal to the rank of the data matrix or not. Our study reveals that the optimal attack strategy depends on the energy budget and the singular values of the data matrix. Specifically, in the scenario where the dimension of the selected subspace is the same as the rank of the data matrix, we show that the optimal rank-one strategy depends solely on the energy budget and the smallest singular value of the data matrix. In the scenario where the dimension of the selected subspace is less than the rank of the original data matrix, the optimal strategy depends not only on the energy budget but also on the $k$th and ($k+1$)th singular values, where $k$ is the dimension of the selected subspace.

Relying on the insights gained from the rank-one case, we then extend our study to the more general case where no rank constraint is imposed. Compared with the case with the rank-one constraint, the attacker now has more degrees of freedom to modify the data, which makes the characterization of the optimal attack strategy significantly more challenging. To solve this optimization problem, we first prove that, under the basis of the principal components of the original data matrix, the optimal attack matrix has only few non-zero entries at particular locations. 
This result greatly reduces the complexity of our problem. With the help of this result, we then simplify our problem to an optimization problem with the objective function being ratio of two quadratic functions. To solve this non-convex problem, 
we further convert our optimization problem to a feasibility problem and find the  close-form solution to this problem. Our result shows that the optimal strategy depends on the energy budget and the $k$th and ($k+1$)th singular values of the data matrix. Our analysis shows that, compared with the optimal rank-one strategy, this strategy leads to a larger subspace distance. 



Our study is related to the recent works on adversary machine learning. For example, \cite{jagielski2018manipulating} studies how to change the data to manipulate the result of the regression learning system. \cite{bayraktar2019adversarial} investigates the optimal modification strategy to maximize the inference errors in a multivariate estimation system. In an interesting related work \cite{pimentel2017adversarial}, the authors study how to design an adversarial data sample and add it to the data matrix in order to maximize the Asimov distance between the subspace estimated by PCA from the contaminated data matrix and that from the original data matrix. \cite{pimentel2017adversarial} focuses on the case where the original data matrix is low-rank and the dimension of the selected subspace is equal to the rank of the data matrix. By contrast, we consider a more powerful adversarial setting, where the data matrix is not constrained to being low-rank, the dimension of the selected subspace does not necessarily equal the rank of the data matrix, and the adversary can modify the whole data matrix instead of only adding one data sample. 

The remainder of the paper is organized as follows. In Section~\ref{sec:formulation}, we describe the precise problem formulation. In Section~\ref{sec:rankone}, we investigate the optimal rank-one attack strategy. We generalize our results to the case without the rank constraint in Section~\ref{sec:general}. In Section~\ref{sec:num}, we provide numerical experiments with both synthesized data and real data to illustrate results obtained in this paper. Finally, we offer concluding remarks in Section~\ref{sec:con}.

\section{Problem formulation}\label{sec:formulation}
In this section, we introduce the problem formulation. Given a data matrix $\mathbf{X}=[\mathbf x_1,\mathbf x_2,\cdots,\mathbf x_n]$ with each $\mathbf{x}_i\in \mathbb{R}^d$, our goal is to learn a low-dimension subspace via PCA. In the data matrix $\mathbf{X}$, we assume that all the preprocessing steps (such as data centering and standardization) have been done. In this paper, we consider an adversarial setup in which an adversary will first observe $\mathbf{X}$ and then carefully design a modification (attack) matrix $\mathbf{\Delta X}$ to change $\mathbf X$ to $\hat{\mathbf{X}}=\mathbf X+\mathbf{\Delta X}$. We denote function $g_k(\cdot)$ as the PCA operation that computes the $k$ leading principal components. Furthermore, let $\mathbb{X} =\textrm{span}(g_k(\mathbf{X}))$ be a $k$-dimensional subspace learned from $\mathbf X$ and $\hat{\mathbb{X}}=\textrm{span}(g_k(\hat{\mathbf{X}}))$
be a $k$-dimensional subspace learned from the modified data $\hat{\mathbf X}$. The goal of the adversary is to design the modification matrix $\mathbf{\Delta X}$ so as to make the distance between $\mathbb{X}$ and $\hat{\mathbb{X}}$ as large as possible. To measure such a distance, we use the largest principal angle between $\mathbb{X}$ and $\hat{\mathbb{X}}$ as defined below \cite{golub2012matrix}.
The largest principal angle plays an important role in the subspace classification problem\cite{huang2015role}. It is closely related to the projection 2-norm which is widely used in engineering applications\cite{he97detectiongapmetric,golub2012matrix,absil2006largest}.
The projection 2-norm also provides a way to measure the discrepancy of the projections of a vector on two distinct subspaces. It is useful in the robustness analysis of the principal component regression (PCR), as one is actually projecting the response value vector onto the selected feature subspace in PCR. We will provide an example to illustrate it in Section~\ref{sec:num} using real data.

\begin{definition}
	Let $\mathbb{X}$ and $\hat{\mathbb{X}}$ be two $k$-dimensional subspaces in $\mathbb{R}^d$. The principal angles $\{\theta_i\}_{i=1}^k$ are defined recursively:
	\begin{align*}
	\cos(\theta_i) =& \max_{\mathbf u_i\in \mathbb{X},\mathbf v_i\in \hat{\mathbb{X}}} \quad  \mathbf u_i^\top \mathbf v_i \\
	\text{s.t.}
	\quad & \|\mathbf u_i\| = \|\mathbf v_i\| = 1, \\
	&\mathbf u_j^\top \mathbf u_i = \mathbf v_j^\top \mathbf v_i =0,\forall \, j=1,2,\cdots,i-1. 
	\end{align*}
\end{definition}
In this paper, we will use $\|\cdot\|$ to denote the $\ell_2$ norm and  $\theta\big(g_k(\mathbf{X}),g_k(\hat{\mathbf{X}})\big)$ or simply $\theta$ to denote the Asimov distance between the subspace $\mathbb{X}$ estimated from $\mathbf{X}$ and the subspace $\hat{\mathbb{X}}$ estimated from $\hat{\mathbf{X}}$. Given an orthonormal basis $\mathbf U_{\mathbb{X}}$ of $\mathbb{X}$ and an orthonormal basis $\mathbf U_{\hat{\mathbb{X}}}$ of $\hat{\mathbb{X}}$, $\{\cos(\theta_1),\cdots,\cos(\theta_k)\}$ are the singular values of $\mathbf U_{\mathbb{X}}^\top \mathbf U_{\hat{\mathbb{X}}}$~\cite{golub2012matrix}. Hence, the Asimov distance is determined by the smallest singular value of $\mathbf U_{\mathbb{X}}^\top \mathbf U_{\hat{\mathbb{X}}}$. It is easy to see that, if no constraint is imposed on $\mathbf{\Delta X}$, $\hat{\mathbf X}$ can be arbitrary and $\theta$ can be easily made to be $\pi/2$. Therefore, we impose an energy constraint on $\mathbf{\Delta X}$. In particular, we assume that the energy of $\mathbf{\Delta X}$ is less than or equal to $\eta$. In this paper, we use the Frobenius norm $\|\mathbf{\Delta X} \|_{\text{F}}$ to measure the energy. Hence, the goal of this attacker is to solve the following optimization problem:

\begin{align}
	\max_{\mathbf{\Delta X}\in \mathbb{R}^{d\times n}}:\quad & \theta\big(g_k(\mathbf{X}),g_k(\hat{\mathbf{X}})\big)
	\label{opt:generaloriginal}\\\nonumber
	\text{s.t.}\quad  & \hat{\mathbf{X}} = \mathbf{X} + \mathbf{\Delta X}, \\ \nonumber
	\quad & \|\mathbf{\Delta X}\|_{\text{F}} \le \eta. 
\end{align}

Even though \eqref{opt:generaloriginal} is a complicated non-convex optimization problem, we will fully characterize the optimal solution to~\eqref{opt:generaloriginal} for any given $\eta$. This characterization will enable us to investigate the impact of this optimal attack with respect to the energy budget $\eta$.


\section{Optimal rank-one adversarial strategy}\label{sec:rankone}
In this section, we will solve~\eqref{opt:generaloriginal} for the special case where the modification matrix $\mathbf{\Delta X}$ is limited to being rank-one. 
The techniques and insights obtained from this special case will be useful for the general case considered in Section~\ref{sec:general}.

With this additional rank-one constraint, $\mathbf{\Delta X}$ can be written as $ \mathbf a\mathbf b^\top$ for some $\mathbf a\in \mathbb{R}^d$ and $\mathbf b\in \mathbb{R}^n$, and the optimization problem~\eqref{opt:generaloriginal} becomes
\begin{align} 
    \max_{\mathbf a\in \mathbb{R}^d,\mathbf b\in \mathbb{R}^n}:\quad & \theta\big(g_k(\mathbf{X}),g_k(\hat{\mathbf{X}})\big)     \label{opt:original} \\ \nonumber
    \text{s.t.}\quad & \hat{\mathbf X} = \mathbf X + \mathbf{\Delta X},\\ \nonumber
    &\mathbf{\Delta X} = \mathbf a\mathbf b^\top,
    & \|\mathbf{\Delta X} \|_{\text{F}} \le \eta.
\end{align}
It is easy to see that, for any feasible solution $(\tilde{\mathbf{a}},\tilde{\mathbf{b}})$ with $||\tilde{\mathbf{b}}||\neq 1$, we can construct another feasible solution $(||\tilde{\mathbf{b}}||\tilde{\mathbf{a}},\tilde{\mathbf{b}}/||\tilde{\mathbf{b}}||)$ that gives the same objective function value. Hence, without loss of optimality, we will fix the norm of $\mathbf b$ to be $1$ throughout this section. 


Based on the value of $k$, i.e., the dimension of the subspace we select, we will first present the solution to the case when $k=\text{rank}(\mathbf X)$, and then generalize the result to the case when $k<\text{rank}(\mathbf X)$.

\subsection{Case with $k=\text{rank}(\mathbf X)$}\label{sec:kerank}
In this subsection, we consider the case when the dimension of the subspace selected is equal to the rank of the data matrix. In this case, the span of $\mathbf{X}$ equals the span of $g_k(\mathbf{X})$.
Furthermore, we divide this case into two scenarios where the data matrix is full-rank and the data matrix is low-rank. 
\subsubsection{Full-Rank Case}
In the full column rank case, $\text{rank}(\mathbf X)=n$, where $n \leq  d$.  This case arises when the number of samples is limited, for example, at the beginning of online PCA. In this case, the span of $\hat{\mathbf{X}}$ is equal to the span of $g_k(\hat{\mathbf{X}})$,
and hence we can write $\theta\big(g_k(\mathbf{X}),g_k(\hat{\mathbf{X}})\big)$ as $\theta(\mathbf{X},\hat{\mathbf{X}})$. In the following, we first find the expression of $\theta(\mathbf{X},\hat{\mathbf X})$ for any given $\hat{\mathbf X}= \mathbf X +\mathbf{a}\mathbf{b}^T$. Using this expression, we then characterize the optimal attack matrix $\mathbf \Delta \mathbf X$.

Suppose the compact SVD of $\mathbf X$ is 
$ \mathbf X = \mathbf U \mathbf\Sigma \mathbf V^\top =\mathbf U \mathbf W,$
where $\mathbf{\Sigma} = \text{diag}(\sigma_1,\sigma_2,\cdots,\sigma_n)$. 
One set of orthonormal bases for the column space of $\mathbf X$ is $\mathbf U$. We can also use SVD to find a set of orthonormal bases $\tilde{\mathbf U}$ of $\text{span}(\hat{\mathbf X})$. 

Since $\hat{\mathbf X} = \mathbf{X} + \mathbf a\mathbf b^\top $, $\tilde{\mathbf U}$ can be directly expressed as a function of $\mathbf U$ \cite{zimmermann2017closed}:
$$\tilde{\mathbf U} = \mathbf U + (\alpha \mathbf U \mathbf w+\beta \mathbf s)\mathbf w^\top,$$
where 
\begin{align*}
&\mathbf a_{u^\perp}=(\mathbf I - \mathbf U\mathbf U^\top)\mathbf a,
\, &&\mathbf s= \mathbf a_{u^\perp}/\|\mathbf a_{u^\perp}\|, \\
&\tilde{\mathbf w}= -\mathbf W^{-\top}\mathbf b, 
\, &&\mathbf w= \tilde{\mathbf w}/\|\tilde{\mathbf w}\|,\\ 
&\omega = (1-\mathbf a^\top \mathbf U\tilde{\mathbf w})/\|\mathbf a_{u^\perp}\|, 
\, &&\mathbf g= [\tilde{\mathbf w},\omega]^\top, \\
&\alpha = |\omega|/\|\mathbf g\|-1,
\, &&\beta=-\text{sign}(\omega)\|\tilde{\mathbf w}\|/\|\mathbf g\|,
\end{align*}
and $\mathbf W^{-\top}=\mathbf (\mathbf W^{-1})^\top$. Hence, we have 
$\mathbf U^\top \tilde{\mathbf U} 
= \mathbf U^\top\left(\mathbf U+(\alpha \mathbf U \mathbf w+\beta \mathbf s)\mathbf w^\top \right)
=\mathbf I + \alpha \mathbf w\mathbf w^\top. $
The singular values of $\mathbf I + \alpha \mathbf w\mathbf w^\top$ are $\{1,1,\cdots,1+\alpha \mathbf w^\top \mathbf w\}$. Since $\mathbf w^\top \mathbf w=1, 1+\alpha = |\omega|/\|\mathbf g\|$, the smallest singular value of $\mathbf U^\top \tilde{\mathbf U} $ is $\cos(\theta) = |\omega|/\|\mathbf g\|$. Our objective is to maximize $\theta$, which is equivalent to minimizing the smallest singular value of $\mathbf U^\top \tilde{\mathbf U} $. Hence, the optimization problem~\eqref{opt:original} is simplified as 
\begin{align*}
    \min_{\mathbf a,\mathbf b}:\quad&|\omega|/\|\mathbf g\| \\
    \text{s.t.}\quad &  \|\mathbf a \mathbf b^\top\|_{\text{F}}=\|\mathbf a\|\|\mathbf b\| \le \eta,
\end{align*}
where we use the identity $\|\mathbf a\|\|\mathbf b\| = \|\mathbf a \cdot \mathbf b^\top \|_{\text{F}}$. Expanding the objective function, we have 
\begin{eqnarray}\frac{|\omega|}{\|\mathbf g\|} 
=\frac{|1+\mathbf a_u^\top \mathbf W^{-\top}\mathbf b|}
{\|[\|\mathbf a_{u^\perp}\|\mathbf W^{-\top}\mathbf b,1+\mathbf a_{u}^\top \mathbf W^{-\top}\mathbf b]\|},\label{eq:cost}\end{eqnarray}
where $\mathbf a_u = \mathbf U^\top \mathbf a$.

Since $\mathbf W = \mathbf \Sigma \mathbf V^\top $, we have $\mathbf W^{-\top}\mathbf b = \mathbf \Sigma^{-1}\mathbf V^\top \mathbf b$. As $\mathbf V$ is a unitary matrix, changing the coordinate $\mathbf b \Leftarrow \mathbf V^\top \mathbf b$ does not change the constraint. The value $\mathbf a_u^\top \mathbf W^{-\top}\mathbf b$ in the original coordinate is the same as $\mathbf a_u^\top \mathbf \Sigma^{-1}\mathbf b$ in the new coordinate. In the following, we will use this new coordinate system and the cost function in~\eqref{eq:cost} can be written as
\begin{eqnarray}\frac{|\omega|}{\|\mathbf g\|} 
=\frac{|1+\mathbf a_u^\top \mathbf \Sigma^{-1}\mathbf b|}
{\|[\|\mathbf a_{u^\perp}\|\mathbf\Sigma^{-1}\mathbf b,1+\mathbf a_{u}^\top \mathbf \Sigma^{-1}\mathbf b]\|}.\label{eq:costnew}\end{eqnarray}

The objective function~\eqref{eq:costnew} is zero if and only if the numerator is zero. Using the matrix norm inequality\cite{horn2012matrix},
we have
\begin{align*}
|\mathbf a_u^\top \mathbf \Sigma^{-1}\mathbf b| 
&\le \|\mathbf a_u\| \|\mathbf b\|\|\mathbf \Sigma^{-1}\|_{2}
=\frac{1}{\sigma_n} \|\mathbf a_u\| \|\mathbf b\|\\
&\overset{(a)}{\le}\frac{1}{\sigma_n} \|\mathbf a\| \|\mathbf b\| 
=\frac{1}{\mathbf \sigma_n} \|\mathbf a \mathbf b^\top\|_{\text{F}}\overset{(b)}{\leq} \frac{\eta}{\mathbf \sigma_n} ,
\end{align*}
where $\|\mathbf{\Sigma}^{-1}\|_2$ is the induced 2-norm of matrix $\mathbf{\Sigma}^{-1}$, in (a) we use $\|\mathbf a_u\| \le \|\mathbf a\|$, and (b) is due to the energy constraint. From the inequalities, we conclude that when $\eta < \sigma_n$, we can not make the numerator to be zero. We now consider two different cases depending on whether we can make the numerator to be zero or not.

\noindent\textbf{Case 1}:  
	When $\eta > \sigma_n$, if we set 
$$\mathbf a_u = [0,0,\cdots,-\sigma_n]^\top, \quad  \mathbf{b}=[0,0,\cdots,1]^\top,$$
and any $\|\mathbf{a}_{u^\perp}\|^2 = \hat{a}^2$ with $0<\hat{a}^2 < \eta^2 - \sigma_n^2$, the numerator will be zero. Since $\mathbf a = \mathbf U\mathbf a_u + (\mathbf I - \mathbf U\mathbf U^\top)\mathbf a_{u^\perp}$, the attacker can make the Asimov distance to be $\pi/2$ by setting:
\begin{align}
    \mathbf a =-\sigma_n \mathbf u_n + \hat{a}\mathbf u_q,  \quad   \mathbf b =\mathbf v_n,
\end{align}
where $\mathbf u_q$ is any vector orthogonal to the column space of $\mathbf X$ and $0<\hat{a}^2 < \eta^2-\sigma_n^2$.

\noindent\textbf{Case 2:}
When $\eta \le \sigma_n$, the value of $1+\mathbf a_u^\top  \mathbf \Sigma^{-1} \mathbf b$ can not reach zero. In this case, it is easy to check that minimizing~\eqref{eq:costnew} is equivalent to maximizing
\begin{equation}
 \frac{\|\mathbf a_{u^\perp}\|^2\|\mathbf \Sigma^{-1}\mathbf b\|^2}{(1+\mathbf a_u^\top \mathbf \Sigma^{-1}\mathbf b)^2}.
\label{obj:simple}
\end{equation}
As $\|\mathbf{b}\|=1$, $\|\mathbf \Sigma^{-1}\mathbf b\|^2$ is maximized when $\mathbf b =[0,0,\cdots,1]^\top$. Furthermore, for any fixed norm of $\mathbf a_u$, $(1+\mathbf a_u^\top \mathbf \Sigma^{-1}\mathbf b)^2$ is minimized when $\mathbf a_u =[0,0,\cdots,-\|\mathbf a_u\|]^\top$, $\mathbf b=[0,0,\cdots,1]^\top$. Hence, for fixed norms of $\mathbf a_u$, $\mathbf a_{u^\perp}$, the objective function~\eqref{obj:simple} is maximized when
\begin{eqnarray}\mathbf a_u=[0,0,\cdots,-\|\mathbf a_u\|]^\top, \quad \mathbf b=[0,0,\cdots,1]^\top.\label{eq:aub}\end{eqnarray}
Let $c=\|\mathbf a_{u^\perp}\|, h=\|\mathbf a_u\|$. Using the optimal form of $\mathbf a_u$ and $\mathbf b$ in~\eqref{eq:aub}, the objective function~\eqref{obj:simple} can be simplified to
\begin{align}\nonumber 
    \max_{c,h}: \quad& \frac{c^2/\sigma_n^2}{(1-h/\sigma_n)^2} \\
    \text{s.t.} \quad & (c^2+h^2) \le \eta^2,
    \label{opt:simple2}
\end{align}
It is easy to check that the objective function is maximized when $c^2+h^2=\eta^2$. 
Hence, we have $c^2=\eta^2-h^2$. Inserting this value of $c$ into the objective function and setting the derivative with respect to $h$ to be $0$, we get a unique solution $h =\eta^2/\sigma_n$. 
At this value of $h$, the second derivative is $\frac{-2\sigma_n^2}{(\sigma_n^2-\eta^2)^3}$, which is negative. It indicates that $h=\eta^2/\sigma_n$ is indeed the maximum point. 
Hence $c=\pm\eta\sqrt{1-\eta^2/\sigma_n^2}$. This implies that the optimal solution to problem~\eqref{opt:original} for Case 2 is 
\begin{align*}
    \mathbf a =-\eta^2/\sigma_n\mathbf u_n \pm \eta\sqrt{1-\eta^2/\sigma_n^2}\mathbf u_q, 
    \quad \mathbf b = \mathbf v_n.
\end{align*}

Summarizing the discussion above, we have the following proposition regarding the optimal value of problem~\eqref{opt:original} in the full-rank case.
\begin{proposition}\label{prop:fullrank}
In the full rank case, the optimal value of~\eqref{opt:original} is 
\begin{align*}
\theta^* =
\begin{cases}
\pi/2, \, &\mbox{if } \eta > \sigma_n \\
\arcsin{(\eta/\sigma_n)}, \,&\mbox{if } \eta \le \sigma_n
\end{cases}.
\end{align*}
\end{proposition}

\subsubsection{Low-Rank Case}
We now consider the case where $\mathbf{X}$ is not full rank. Let $k<\text{min}(d,n)$ be the rank of $\mathbf X$. In this subsection, with a slight abuse of notation, we write the full SVD of $\mathbf{X}$ as $\mathbf X = \mathbf U \mathbf \Sigma \mathbf V^\top$.
The optimal attack matrix could be found by solving
\begin{align}
    \max_{\mathbf a\in \mathbb{R}^d,\mathbf b\in \mathbb{R}^n}:\quad & \mathcal{\theta}\big(\mathbf X,g_k(\hat{\mathbf{X}})\big)\label{opt:low_rank_1} \\ \nonumber 
    \text{s.t.}\quad &
    \hat{\mathbf X} = \mathbf X + \mathbf a\mathbf b^\top, \\
    &\|\mathbf a\|\|\mathbf b\| \le \eta.\nonumber     
\end{align}
We can further simplify this optimization problem as
\begin{align} 
    \max_{\mathbf a\in \mathbb{R}^{k+1}, \mathbf b\in \mathbb{R}^{k+1}}:\quad & \theta\big(\tilde{\mathbf \Sigma}, g_k(\mathbf Y)\big) \label{opt:low_rank_3} \\\nonumber 
    \text{s.t.}\quad & \mathbf Y = \tilde{\mathbf \Sigma} + \mathbf a\mathbf b^\top, \\
    &\|\mathbf a\|\|\mathbf b\| \le \eta,  \nonumber     
\end{align}
where $\tilde{\mathbf \Sigma} =\text{diag}(\sigma_1,\sigma_2,\cdots,\sigma_k,0)$ and $\{\sigma_1,\sigma_2,\cdots,\sigma_k\}$ are singular values of $\mathbf X$.  
Detailed proof of the equivalence between~\eqref{opt:low_rank_1} and~\eqref{opt:low_rank_3} can be found in Appendix~\ref{app:eq12-eq13}. Here, we describe the main idea of the proof. The main step of the simplification is to left multiply the unitary matrix $\mathbf U^\top$ and right multiply the unitary matrix $\mathbf V$ on both $\mathbf X$ and $\Hat{\mathbf X}$. Note that multiplying a unitary matrix does not change the column space and its singular values. In addition, a rank-one modification can only add at most one principal component orthogonal to its original column subspace. Hence, by changing the coordinates, $\mathbf a$ and $\mathbf b$ are $k+1$ dimensional vectors.  

To solve problem~\eqref{opt:low_rank_3}, we divide it into two cases based on the value of the energy budget. 

\noindent\textbf{Case 1}:
When $\eta > \sigma_k$, it is simple to verify that the solution \\
$\mathbf a=[0,0,\cdots,\eta]^\top$, $\mathbf b=[0,0,\cdots,1]^\top$ leads to the maximal Asimov distance, which is $\pi/2$. 

\noindent\textbf{Case 2}:
When $\eta \le  \sigma_k$, the following theorem characterizes the form of optimal $\mathbf a$ and $\mathbf b$. 

\begin{theorem}\label{thm:ab}
There exists an optimal solution to problem~\eqref{opt:low_rank_3} in the following form \begin{align}
    \mathbf a  =[0,\cdots,0,a_k,a_{k+1}]^\top,
    \mathbf b  =[0,0,\cdots,0,1,0]^\top, 
    \label{opt:ab}
\end{align}
with $a_k^2 + a_{k+1}^2 = \eta^2$. 
\end{theorem}
\begin{proof}
Please see Appendix \ref{app:thm1}. 
\end{proof}
In the following, we will find the optimal values of $a_k$ and $a_{k+1}$. Since $\|\mathbf a\|^2 =\eta^2$ and $\mathbf a$ is in the form of \eqref{opt:ab}, we can write $\mathbf a =\eta[0,0,\cdots,\cos(\alpha),\sin(\alpha)]^\top$, where $\alpha \in [0,2\pi)$. To compute the $k$ leading principal components of $\mathbf Y$, we can perform the eigenvalue decomposition of $\mathbf Y \mathbf Y^\top$,
\begin{equation*}
    \mathbf Y\mathbf Y^\top =
    \begin{bmatrix}
    \mathbf \Lambda^2_{k-1} & \mathbf 0 \\
    \mathbf 0         & \mathbf c \mathbf c^\top
    \end{bmatrix}, 
\end{equation*}
where $\mathbf c =[\sigma_k+\eta\cos{\alpha},\eta \sin(\alpha)]^\top$, $\mathbf \Lambda_{k-1} = \text{diag}(\sigma_1,\sigma_2,\cdots,\sigma_{k-1})$. 
Suppose the compact SVD of $\mathbf Y\mathbf Y^\top $ is 
$\mathbf Y\mathbf Y^\top = \Hat{\mathbf U}\Hat{\mathbf \Sigma}\Hat{\mathbf V}^\top,$ where 
\begin{equation*}
\Hat{\mathbf U} =
\begin{bmatrix}
\mathbf I_{k-1} & \mathbf 0\\
\mathbf 0 & \mathbf z
\end{bmatrix},
\end{equation*}
and $\mathbf z\in \mathbb{R}^2$ is the eigenvector of $\mathbf c\mathbf c^\top$ corresponding to its nonzero eigenvalue. Since one orthonormal basis of $\text{span}(\tilde{\mathbf{\Sigma}})$ is $[\mathbf I_k,\mathbf 0]^\top$, the Asimov distance is determined by the singular values of 
\begin{equation*}
    \begin{bmatrix}
    \mathbf I_k \\
    \mathbf 0 
    \end{bmatrix}^\top 
    \cdot
    \begin{bmatrix}
    \mathbf I_{k-1} & \mathbf 0 \\
    \mathbf 0 & \mathbf z
    \end{bmatrix}
    =
    \begin{bmatrix}
    \mathbf I_{k-1} &\mathbf 0 \\
    \mathbf 0 & z_1
    \end{bmatrix}.
\end{equation*}
Hence, the Asimov distance is $\arccos(|z_1|)$. Since $\mathbf c$ is the eigenvector of $\mathbf c\mathbf c^\top$ corresponding to its nonzero eigenvalue, we have 
$|z_1| = \frac{|c_1|}{\|\mathbf c\|}$. 
Our objective function is reduced to 
\begin{align}\label{eq:equivalent}
    \min_{\alpha \in [0,2\pi)}: \quad 
    \frac{|\sigma_k + \eta \cos(\alpha)|}{\|[\sigma_k+\eta\cos(\alpha),\eta \sin(\alpha)]\|}.
\end{align}
It is simple to show that the optimal solution to~\eqref{eq:equivalent} is 
\begin{equation}
\alpha^* = \arccos(- \eta/\sigma_k)
\label{opt_alpha1}
\end{equation}
or 
\begin{equation}
\alpha^* = 2\pi - \arccos(- \eta/\sigma_k).
\label{opt_alpha2}
\end{equation}
Substitute the optimal solution of $\alpha^*$ in~\eqref{opt_alpha1} or~\eqref{opt_alpha2} into the objective of problem~\eqref{eq:equivalent}, we have
$
\sin(\theta^*) = \eta/\sigma_k.
$
Hence, the optimal solution to problem~\eqref{opt:low_rank_3} is
\begin{align*}
    \mathbf a&=\left[0,0,\cdots,-\eta^2/\sigma_k,\pm \eta \sqrt{1-\eta^2/\sigma_k^2}\right]^\top, \\
    \mathbf b&=[0,0,\cdots,0,1,0]^\top,
\end{align*}
which indicates that the optimal solution to problem~\eqref{opt:low_rank_1} is 
\begin{align*}
    \mathbf a  = -\eta^2/\sigma_k \mathbf u_k \pm \eta
    \sqrt{1-\eta^2/\sigma_k^2}\mathbf u_q,\quad 
    \mathbf b =\mathbf v_k, 
\end{align*}
where $\mathbf u_q$ is any vector orthogonal to the column space of $\mathbf X$. 
The corresponding optimal subspace distance is
 $   \theta^* = \arcsin(\eta/\sigma_k).$
In summary, we have
\begin{proposition}
The optimal Asimov distance in the low-rank case is
\begin{align}
\theta^* =
\begin{cases}
\pi/2, \, &\mbox{if } \eta > \sigma_k \\
\arcsin{(\eta/\sigma_k)}, \,&\mbox{if } \eta \le \sigma_k
\end{cases}.
\label{opt:rank-one-theta}
\end{align}
\end{proposition}

The result is similar to the full column rank case characterized in Proposition~\ref{prop:fullrank}. 

\subsection{Case with $k<\text{rank}(\mathbf{X})$}\label{sec:knerank}


In this section, we consider the more practical but much more challenging case with $k<\text{rank}(\mathbf{X})$. 

Given the data matrix $\mathbf X \in \mathbb{R}^{d\times n}$, without loss of generality, we assume $d \le n$ and $\text{rank}(\mathbf X)=d$. Assume the full SVD of $\mathbf{X}$ is $\mathbf X = \mathbf U \mathbf{\Sigma}\mathbf{V}$, where $\mathbf{U}\in \mathbb{R}^{d\times d}$, $\mathbf{\Sigma}\in \mathbb{R}^{d\times n}$, $\mathbf{V}\in \mathbb{R}^{n\times n}$, and the singular values of $\mathbf{X}$ are $\{\sigma_1,\sigma_2,\cdots,\sigma_k,\cdots,\sigma_d\}$. Recall that we denote $g_k(\cdot)$ as the PCA  operation that computes the $k$ leading principal components. In this scenario, as the original data matrix is not low-rank, we will perform PCA both on the original data matrix and on the modified data matrix. Hence, the optimal rank-one modification matrix can be found by solving the following optimization problem
\begin{align}
    \max_{\mathbf{a} \in \mathbb{R}^d,\mathbf{b}\in \mathbb{R}^n}:\quad &\theta\big(g_k(\mathbf X),g_k(\hat{\mathbf X})\big) 
    \label{opt:general-original}\\\nonumber
    \text{s.t.}\quad & \hat{\mathbf{X}} = \mathbf{X} + \mathbf{a}\mathbf{b}^\top, \\ \nonumber
    & \|\mathbf{a}\mathbf{b}^\top\|_{\text{F}} \le \eta. 
\end{align}
By diagonalizing the data matrix and using similar arguments in Appendix~\ref{app:eq12-eq13}, \eqref{opt:general-original} can be further simplified as 
\begin{align}
    \max_{\mathbf{a}\in \mathbb{R}^d,\mathbf{b}\in \mathbb{R}^n}:
    \quad &\theta\big(g_k(\mathbf{\Sigma}),g_k(\mathbf{Y})\big)
    \label{opt:general-simplified}\\ \nonumber 
    \text{s.t.}\quad & \mathbf{Y} = \mathbf{\Sigma} + \mathbf{a}\mathbf{b}^\top, \\ \nonumber
    &\|\mathbf{a}\mathbf{b}^\top\|_{\text{F}}\le \eta,
\end{align}
where $g_k(\mathbf{\Sigma})=[\mathbf{I}_k,\mathbf{0}]^\top\in \mathbb{R}^{d\times k}$. Here we also perform variable change $\mathbf a \Leftarrow \mathbf{U}^\top \mathbf a$ and $\mathbf{b} \Leftarrow \mathbf{V}^\top \mathbf{b}$. To solve this optimization problem, we divide it into two cases depending on the energy budget and the difference between $\sigma_k$ and $\sigma_{k+1}$. 

\noindent\textbf{Case 1}: When $\eta\ge \sigma_k - \sigma_{k+1}$, we have one simple solution $\mathbf{a}=[0,0,\cdots,0,\eta,0,\cdots,0]^\top$, where $\eta$ is in the ($k+1$)th coordinate, and $\mathbf{b}=[0,0,\cdots,0,1,0,\cdots,0]^\top$, where element $1$ is in the ($k+1$)th coordinate. Clearly, this setting of $\mathbf{a}$ and $\mathbf{b}$ leads to the maximal subspace distance, which is $\pi/2$. 

\noindent\textbf{Case 2}: When $\eta < \sigma_k-\sigma_{k+1}$, the following theorem gives the form of the optimal solution. 
\begin{theorem}\label{theorem:rank-one-optimality}
The optimal solution to problem~\eqref{opt:general-simplified} should be in the form of
\begin{align} 
    \mathbf{a} &=[0,0,\cdots,a_k,a_{k+1},0,\cdots,0]^\top ,
    \label{opt:a}\\
    \mathbf{b} 
    &= [0,0,\cdots,b_k,b_{k+1},0,\cdots,0]^\top ,
    \label{opt:b}
\end{align}
where $a_k^2+a_{k+1}^2 =\eta^2$ and $b_k^2+b_{k+1}^2 = 1$. 
\end{theorem}
\begin{proof}
Please see Appendix \ref{app:thm2} for details.
\end{proof}

As the optimal solution of $\mathbf{a}$ and $\mathbf{b}$ are in the form of~\eqref{opt:a} and~\eqref{opt:b}, we can parametrize $\mathbf{a}$ and $\mathbf{b}$ with parameters $\alpha$ and $\beta$ using $\mathbf{a} = \eta[0,0,\cdots,\cos(\alpha),\sin(\alpha),0,\cdots,0]^\top$ and $\mathbf{b} = [0,0,\cdots,\cos(\beta),\sin(\beta),0,\cdots,0]^\top$ respectively. 

As a result, the modified data matrix $\mathbf{Y}$ can be written as 
\begin{align}
    \nonumber 
    \mathbf{Y}&=
    \begin{bmatrix}
    \mathbf{\Sigma}_1 & \mathbf{0}& \mathbf{0} &\mathbf{0}\\
    \mathbf{0}   &\mathbf{\Sigma}_2 &\mathbf{0} &\mathbf{0}\\
    \mathbf{0}   &\mathbf{0} &\mathbf{\Sigma}_3& \mathbf{0}\\
    \end{bmatrix},
\end{align}
where $\mathbf{\Sigma}_1 = \text{diag}(\sigma_1,\sigma_2,\cdots,\sigma_{k-1})$, $\mathbf{\Sigma}_3=\text{diag}(\sigma_{k+2},\cdots,\sigma_d)$, and 
\begin{align}
\mathbf{\Sigma}_2=
    \begin{bmatrix}
     \sigma_k + \eta\cos(\alpha)\cos(\beta)  &\eta\cos(\alpha)\sin(\beta) \\
     \eta\sin(\alpha)\cos(\beta) &\sigma_{k+1}+\eta\sin(\alpha)\sin(\beta)
    \end{bmatrix}.
    \label{Sigma2}
\end{align}
Since $\mathbf{Y}$ has the pseudo block diagonal form,
the singular values and principal components of $\mathbf{Y}$ are determined by the SVD of $\mathbf{\Sigma}_1$, $\mathbf{\Sigma_2}$, and $\mathbf{\Sigma}_3$. For notation convenience, we denote  $\mathbf{\Sigma_2}=\mathbf{D}+\eta\bar{\mathbf a}\bar{\mathbf b}^\top$, where $\mathbf{D}=\text{diag}(\sigma_k,\sigma_{k+1})$, $\bar{\mathbf{a}}=[\cos{\alpha},\sin{\alpha}]^\top$, and $\bar{\mathbf{b}}=[\cos{\beta},\sin{\beta}]^\top$. Let $\xi_1$ and $\xi_2$ be the two singular values of $\mathbf{\Sigma}_2$ and denote their corresponding left singular vectors as
\begin{align}
    \mathbf{W} = [\mathbf{w}_1\, \mathbf{w}_2]=
    \begin{bmatrix}
    \cos \varphi & -\sin\varphi \\
    \sin \varphi & \cos\varphi
    \end{bmatrix}.
    \label{eig:W}
\end{align}
The following lemma characterizes the form of $k$-dimensional subspace learned by PCA from $\mathbf{Y}$.
\begin{lemma}
\begin{align*}
  g_k(\mathbf{Y}) = 
  \begin{bmatrix}
  \mathbf{I}_{k-1} & \mathbf{0}\\
  \mathbf{0}       & \mathbf{w}_1\\
  \mathbf{0}       & \mathbf{0}
  \end{bmatrix}.
\end{align*}
\end{lemma}
\begin{proof}
According to the perturbation theory~\cite{pertRank}, the singular values of $\mathbf{\Sigma}_2$ must satisfy
$$\xi_2<\sigma_{k}, \quad \xi_1>\sigma_{k+1}.$$
It indicates that $\xi_1>\sigma_k$, $\xi_2>\sigma_k$ and $\xi_1<\sigma_{k+1}$,  $\xi_2<\sigma_{k+1}$ will not happen. Hence, we will select the eigenvector corresponding to singular value $\xi_1$ as one of the leading $k$ principal components, which completes the proof.
\end{proof}

Since one set of orthonormal bases for $g_k(\mathbf{\Sigma})$ is $[\mathbf{I}_{k},\mathbf{0}]^\top$, then the subspace distance $\theta\big(g_k(\mathbf{\Sigma}),g_k(\mathbf{Y})\big)$ is determined by the singular values of 
\begin{align*}
    \begin{bmatrix}
    \mathbf{I}_k\\ 
    \\
    \mathbf{0}
    \end{bmatrix}^\top
    \cdot 
    \begin{bmatrix}
    \mathbf{I}_{k-1}&\mathbf{0}\\
    \mathbf{0}&\mathbf{w}_1\\
    \mathbf{0}&\mathbf{0}
    \end{bmatrix} 
    =\text{diag}(1,1,\cdots,\cos \varphi).
\end{align*}
Hence, the subspace distance is $\arccos(|\cos \varphi|)$ and our optimization problem can be equivalently formulated as 
\begin{align}
    \min_{\alpha\in [0,2\pi),\beta\in [0,2\pi)}\quad & |\cos\varphi|.
    \label{opt:origi}
\end{align}
Let $\mathbf{Z} = \mathbf{\Sigma}_2\mathbf{\Sigma}_2^\top$, we can compute $\mathbf{W}$ through eigenvalue decomposition of $\mathbf{Z}$.
According to the equality $\mathbf{\Sigma}_2\mathbf{\Sigma}_2^\top=\mathbf{W}\cdot\text{diag}(\xi_1^2,\xi_2^2)\cdot\mathbf{W}^\top$, we have
\begin{align*}
    \mathbf{Z} = &\begin{bmatrix}
    Z_{1,1} & Z_{1,2}\\
    Z_{2,1} & Z_{2,2}
    \end{bmatrix}\\ 
    =&
    \begin{bmatrix}
     \xi_1^2 \cos^2\varphi+\xi_2^2 \sin^2\varphi & (\xi_1^2-\xi_2^2)\cos\varphi\sin\varphi \\
     (\xi_1^2-\xi_2^2)\cos\varphi\sin\varphi &
     \xi_1^2\sin^2\varphi+\xi_2^2\cos^2\varphi
    \end{bmatrix}.\\
\end{align*}
From this equation, we obtain
\begin{align*}
    \begin{cases}
      \cos(2\varphi)(\xi_1^2-\xi_2^2) = Z_{1,1} - Z_{2,2} \\
      \sin(2\varphi)(\xi_1^2-\xi_2^2) = Z_{1,2} + Z_{2,1}
    \end{cases}.
\end{align*}
Then we can compute $\varphi$ through
\begin{equation}
\varphi =0.5\text{atan2}(a_y,a_x),
\label{eq:varphi}
\end{equation}
where $\text{atan2}(\cdot,\cdot)$ is the four-quadrant inverse tangent function, $a_x=Z_{1,1}-Z_{2,2}$, and $a_y=Z_{1,2}+Z_{2,1}$. In our case, the specific expressions of $a_x$ and $a_y$ are
\begin{align}
    \begin{cases}
     a_x =&\sigma_k^2 - \sigma_{k+1}^2 + 2\sigma_k\eta\cos(\alpha)\cos(\beta)\\
     &-2\sigma_{k+1}\eta\sin(\alpha)\sin(\beta)+\eta^2\cos(2\alpha),  \\
     a_y =& 2\eta\Big(\sigma_k\sin(\alpha)\cos(\beta)+\sigma_{k+1}\cos(\alpha)\sin(\beta)\\
     &+\eta\cos(\alpha)\sin(\alpha)\Big).
    \end{cases}
    \label{equ:axay}
\end{align}

Let us write $a_x$ and $a_y$ as a function of $\alpha$ and $\beta$: $a_x=a_x(\alpha,\beta)$ and $a_y=a_y(\alpha,\beta)$. To further restrict the domains of $\alpha$ and $\beta$, we analyze the properties of the angle $\varphi$ in~\eqref{eq:varphi} as a function of $\alpha$ and $\beta$. First, we have $a_x(\alpha,\beta) = a_x(\pi+\alpha,\pi+\beta)$ and $a_y(\alpha,\beta)=a_y(\pi+\alpha,\pi+\beta)$. So $\varphi(\alpha,\beta)=\varphi(\pi+\alpha,\pi+\beta)$. This property indicates that we only need to consider the function value in the domain $\alpha \in [0,\pi], \beta \in [-\pi,\pi]$. Second, $a_x(\alpha,\beta)=a_x(\pi-\alpha,\pi-\beta)$ and $a_y(\alpha,\beta)=-a_y(\pi-\alpha,\pi-\beta)$, and then we have $\varphi(\alpha,\beta)=-\varphi(\pi-\alpha,\pi-\beta)$. Since $\cos(\varphi)$ is an even function, we only need to consider the function with domain $\alpha \in [0,\pi/2],\beta \in [-\pi,\pi]$. Note that $\mathbf{\Sigma}_2$ is in the form of~\eqref{Sigma2}, the variance in the direction of $\mathbf{e}_k$ is $v_k = \cos(\alpha)^2+\sigma_k^2+2\cos(\alpha)\cos(\beta)$, and the variance in the direction of $\mathbf{e}_{k+1}$ is $v_{k+1} =\sin(\alpha)^2+\sigma_{k+1}^2+2\sin(\alpha)\sin(\beta)$. 
To maximize the subspace distance, we should make $v_k$ small and make $v_{k+1}$ large. Apparently, the sign of $\cos(\alpha)\cos(\beta)$ should be negative and the sign of $\sin(\alpha)\sin(\beta)$ should be positive. Hence the optimal $\alpha$ and $\beta$ should satisfy $\alpha \in [0,\pi/2]$ and $\beta \in [\pi/2,\pi]$. As a result, the optimization problem~\eqref{opt:origi} can be written as
\begin{equation}
    \min_{\alpha\in[0,\pi/2],\beta\in[\pi/2,\pi]}:\, |\cos\left(\varphi(\alpha,\beta)\right)|.
    \label{opt:theta2}
\end{equation}

The following theorem characterizes the optimal solution to problem~\eqref{opt:theta2}. 
\begin{theorem}\label{thm:sol-k<r}
The optimal solution to problem~\eqref{opt:theta2} is
\begin{align}
    \begin{cases}
    \alpha^*&=\arccos\left(\sqrt{\frac{\sigma_k^2-\sigma_{k+1}^2+\eta^2 - \sqrt{H}}{2(\sigma_k^2-\sigma_{k+1}^2)}}\right), \\
    \beta^* &=\arccos\left(- \sqrt{\frac{\sigma_k^2-\sigma_{k+1}^2+\eta^2+ \sqrt{H}}{2(\sigma_k^2-\sigma_{k+1}^2)}}\right),
    \end{cases}
    \label{opt:albe}
\end{align}
where $H = \sigma_k^4+\sigma_{k+1}^4+\eta^4-2\sigma_k^2\sigma_{k+1}^2-2\sigma_k^2\eta^2-2\sigma_{k+1}^2\eta^2$.
\end{theorem}
\begin{proof} 
Please see Appendix~\ref{app:sol-kkt1}.
\end{proof}

Accordingly, the optimal solution to  problem~\eqref{opt:general-original} is
\begin{align}
    \mathbf{a}^*&=\eta\cos(\alpha^*)\mathbf{u}_k+\eta\sin(\alpha^*)\mathbf{u}_{k+1},\label{sol:genera-a}\\
    \mathbf{b}^*&=\cos(\beta^*)\mathbf{v}_k+\sin(\beta^*)\mathbf{v}_{k+1}.\label{sol:genera-b}
\end{align}
Furthermore, the optimal subspace distance $\theta^*$ can be computed according to~\eqref{equ:axay} and~\eqref{eq:varphi}. Moreover, according to the properties of the function $\varphi(\alpha,\beta)$ we have discussed before, there are other three optimal solutions
$$(-\alpha^*, -\beta^*),\quad (\pi-\alpha^*,\pi-\beta^*),\quad (\alpha^*-\pi,\beta^*-\pi),$$
which lead to the same optimal objective value.


\section{Optimal adversarial strategy without the rank constraint}\label{sec:general}

Using the insights gained from Section~\ref{sec:rankone}, we now characterize the optimal attack strategy in the general case without the rank-one constraint by solving~\eqref{opt:generaloriginal}. We will directly consider the general case with $k\leq \text{rank}(\mathbf{X})$. 

Following the similar transformation from~\eqref{opt:low_rank_1} to \eqref{opt:low_rank_3}, we can simplify the optimization problem~\eqref{opt:generaloriginal} as
\begin{align}
    \max_{\mathbf{B}\in \mathbb{R}^{d\times n}}:\quad & \theta\big(g_k(\mathbf{\Sigma}),g_k(\mathbf{Y})\big)
    \label{opt:no-rank-simplified} \\ \nonumber
    \text{s.t.}\quad & \mathbf{Y} = \mathbf{\Sigma} + \mathbf{B}, \\\nonumber
    &\|\mathbf{B}\|_{\text{F}} \le \eta,
\end{align}
where without loss of generality we assume $ d\le n$, the full SVD of the data matrix is $\mathbf{X} = \mathbf{U}\mathbf{\Sigma}\mathbf{V}^\top$, the singular values of the data matrix are $\{\sigma_1,\sigma_2,\cdots,\sigma_d\}$, and $\mathbf{B} = \mathbf{U}^\top  \mathbf{\Delta X} \mathbf{V}$. To identify the optimal modification matrix $\mathbf{B}$ in problem~\eqref{opt:no-rank-simplified}, we divide it into two cases.\\
\textbf{Case 1:} When $\eta \ge  \frac{\sigma_k - \sigma_{k+1}}{\sqrt{2}}$, by setting $b_{k,k}=-\eta/\sqrt{2}$, $b_{k+1,k+1} = \eta/\sqrt{2}$, and all other entries of $\mathbf{B}$ to zero, where $b_{i,j}$ is the element in the $i$th row and $j$th column of $\mathbf{B}$, this will lead to the maximal subspace distance, $\pi/2$. \\
\textbf{Case 2:} When $\eta <\frac{\sigma_k - \sigma_{k+1}}{\sqrt{2}}$, the following theorem states the form of the optimal $\mathbf{B}$. 
\begin{theorem} \label{theorem:no-rank-optimality}
The optimal $\mathbf{B}$ to problem~\eqref{opt:no-rank-simplified} has only four possible non-zero entries: $b_{k,k}, b_{k,k+1}, b_{k+1,k}$ and $b_{k+1,k+1}$. 
\end{theorem}
\begin{proof}
Please see Appendix \ref{app:thm3}.
\end{proof}
This characterization reduces the complexity of problem~\eqref{opt:no-rank-simplified}. Using this optimal form of $\mathbf{B}$ and following similar steps leading to~\eqref{eq:varphi}, we can write the subspace distance as 
\begin{align}
    \theta =0.5\left|\text{atan2}(b_y,b_x)\right|,
    \label{eq:no-rank-theta}
\end{align}
where 
\begin{align*}
    b_y &= 2\big((b_{k,k}+\sigma_k)b_{k+1,k}+(b_{k+1,k+1}+\sigma_{k+1})b_{k,k+1}\big), \\
    b_x &= (b_{k,k}+\sigma_k)^2+b_{k,k+1}^2 - (b_{k+1,k+1}+\sigma_{k+1})^2 - b_{k+1,k}^2.
\end{align*}
It is easy to see that we can change the sign of $b_y$ by changing the signs of $b_{k,k+1}$ and $b_{k+1,k}$. We also have $b_x>0$, as 
\begin{align*}
    &\frac{b_x}{\|[b_{k,k}+\sigma_k,\, b_{k,k+1}]\|+\|[b_{k+1,k+1}+\sigma_{k+1},\, b_{k+1,k}]\|}\\
    &= \|[b_{k,k}+\sigma_k,\, b_{k,k+1}]\| -  \|[b_{k+1,k+1}+\sigma_{k+1},\, b_{k+1,k}]\| \\
    &\ge \sigma_k -\sigma_{k+1} - \|[b_{k,k},\,b_{k,k+1}]\| - \|[b_{k+1,k},\,b_{k+1,k+1}]\|\\
    &\ge \sigma_k - \sigma_{k+1} - \sqrt{2}\eta >0.
\end{align*}
Using these two facts and the fact that $\text{atan2}(b_y,b_x)$ is an odd function of $b_y$ when $b_x >0$,  
we know that maximizing $\theta$ in~\eqref{eq:no-rank-theta} is 
equivalent to maximizing $b_y/b_x$. Hence, our optimization problem can be written as 
\begin{align} \label{opt:feasible0}
    \max_{\mathbf{u}}:\quad & \frac{\mathbf{u^\top}\mathbf{A}_1\mathbf{u}}{\mathbf{u}^\top\mathbf{A}_2\mathbf{u}} \\ \nonumber
    \text{s.t.}\quad & \|\mathbf{u} -\boldsymbol{\sigma}\|^2 \le \eta^2,
\end{align}
where $\mathbf{u} \triangleq \bar{\mathbf{b}}+\boldsymbol{\sigma}$ with
$\bar{\mathbf{b}} = [b_{k,k}, b_{k+1,k},b_{k,k+1},b_{k+1,k+1}]^\top$ and
$\boldsymbol{\sigma}=[\sigma_k,0,0,\sigma_{k+1}]^\top$, 
\begin{align*}
    \mathbf{A}_1 = 
    \begin{bmatrix}
    0&1&0&0\\
    1&0&0&0\\
    0&0&0&1\\
    0&0&1&0
    \end{bmatrix}\text{,\quad and} \quad 
    \mathbf{A}_2 = 
    \begin{bmatrix}
    1&0&0&0 \\
    0&-1&0&0 \\
    0&0&1&0\\
    0&0&0&-1
    \end{bmatrix}.
\end{align*}

The objective function is the ratio of two quadratic functions. It is a non-convex problem in general. In the following, we transform this problem into a feasibility problem and obtain the closed-form solution analytically. 

Let $\lambda$ denote the value of the objective function in~\eqref{opt:feasible0}. We can rewrite the optimization problem~\eqref{opt:feasible0} as
\begin{align} \nonumber
    \max_{\lambda, \mathbf{u}}:\quad & \lambda\\
    \text{s.t.}\quad & \frac{\mathbf{u^\top}\mathbf{A}_1\mathbf{u}}{\mathbf{u}^\top\mathbf{A}_2\mathbf{u}} = \lambda 
    \label{opt:no-rank-first-constrain},\\ \nonumber
    &  \|\mathbf{u} -\boldsymbol{\sigma}\|^2 \le \eta^2.
\end{align}
The first constraint can be written as $\mathbf{u}^\top(\mathbf{A}_1-\lambda\mathbf{A}_2)\mathbf{u} = 0$, where 
\begin{align*}
    \begin{bmatrix}
    \mathbf{Q} & \mathbf{0} \\
    \mathbf{0} & \mathbf{Q}
    \end{bmatrix} 
    \triangleq
    \mathbf{A}_1 -\lambda \mathbf{A}_2 = 
    \begin{bmatrix}
    -\lambda & 1 & 0 & 0 \\
    1  & \lambda & 0 & 0 \\
    0  & 0 & -\lambda & 1 \\
    0  & 0 & 1 & \lambda 
    \end{bmatrix}.
\end{align*}
To further simplify the constraint, we perform eigenvalue decomposition on 
$
    \mathbf{Q} = \mathbf{P}
    \mathbf{\Lambda}
    \mathbf{P}^\top, 
$
where $\mathbf{\Lambda} = \text{diag}(\sqrt{\lambda^2+1}, - \sqrt{\lambda^2+1})$ and
\begin{align}
    \mathbf{P} = t
    \begin{bmatrix}
    1 &  -(\sqrt{\lambda^2+1}+\lambda)\\
    \sqrt{\lambda^2+1} + \lambda & 1
    \end{bmatrix},
    \label{eq:no-rank-P}
\end{align}
with $t = 1/\sqrt{(\sqrt{\lambda^2+1}+\lambda)^2+1}$. 

We further perform variable change $\mathbf{v} \triangleq \text{diag}(\mathbf{P}^\top, \mathbf{P}^\top )\mathbf{u}$. Thus, the constraint~\eqref{opt:no-rank-first-constrain} is equivalent to $\mathbf{v}^\top \mathbf{\Lambda}\mathbf{v}=0$, which indicates 
$v_1^2 + v_3^2 = v_2^2 + v_4^2$. With this, the optimization problem is simplified as 
\begin{align}
    \max_{\lambda, \mathbf{v}}:\quad & \lambda \label{opt:no-rank-feasibility} \\
    \text{s.t.}\quad & v_1^2 + v_3^2 = v_2^2 + v_4^2 \label{opt:no-rank-f2},\\
    & \|\mathbf{v} - \bar{\boldsymbol{\sigma}}\|^2 \le \eta^2,
    \label{opt:no-rank-f3}
\end{align}
where $\bar{\boldsymbol{\sigma}} = \text{diag}(\mathbf{P}^\top, \mathbf{P}^\top ) \boldsymbol{\sigma} =[p_{1,1}\sigma_k,\, p_{1,2}\sigma_k,\, p_{2,1}\sigma_{k+1},\, p_{2,2}\sigma_{k+1}]^\top.$ Note that $p_{1,2} = - p_{2,1}$ and $p_{2,2} = p_{1,1}$, we have $\bar{\boldsymbol{\sigma}} = [p_{1,1}\sigma_k,\, -p_{2,1}\sigma_k, \,p_{2,1}\sigma_{k+1}, \,p_{1,1}\sigma_{k+1}]^\top.$ 

Now, problem~\eqref{opt:no-rank-feasibility} can be solved by checking the feasibility of \eqref{opt:no-rank-f2} and \eqref{opt:no-rank-f3} given a particular $\lambda$. Given $\lambda$, the feasibility of problem\eqref{opt:no-rank-feasibility} is equivalent to the feasibility of 
\begin{align}
    \min_{v_1^2+ v_3^2 = v_2^2+ v_4^2} \|\mathbf{v} - \bar{\boldsymbol{\sigma}}\|^2 \le \eta^2. \label{ineq:feasibility}
\end{align}
Note that $\bar{\boldsymbol{\sigma}}$ depends on $\lambda$, we denote the left hand side of inequality~\eqref{ineq:feasibility} as $f(\mathbf{v},\lambda) = \|\mathbf{v} - \bar{\boldsymbol{\sigma}}\|^2$ and parametrize $\mathbf{v}$ as 
\begin{align}
       v_1 = r \cos(\alpha) , 
    v_2 = r \cos(\beta) , 
    v_3 = r \sin(\alpha) ,
    v_4 = r \sin(\beta). 
    \label{eq:no-rank-x}  
\end{align}
It is easy to verify that the minimum point of $f(\mathbf{v},\lambda)$ in terms of $\mathbf{v}$ is obtained at the following stationary point
\begin{align} 
    \begin{cases}
     r = \frac{1}{2}\left(\sqrt{p_{1,1}^2\sigma_k^2+p_{2,1}^2\sigma_{k+1}^2} + \sqrt{p_{2,1}^2\sigma_k^2+p_{1,1}^2\sigma_{k+1}^2}\right), \\
     \cos(\alpha) = p_{1,1}\sigma_k/\sqrt{p_{1,1}^2\sigma_{k}^2+p_{2,1}^2\sigma_{k+1}^2}, \\
    \sin(\alpha) = p_{2,1}\sigma_{k+1}/\sqrt{p_{1,1}^2\sigma_{k}^2+p_{2,1}^2\sigma_{k+1}^2},\\
    \cos(\beta) =-p_{2,1}\sigma_k/\sqrt{p_{1,1}^2\sigma_{k+1}^2+p_{2,1}^2\sigma_k^2} ,\\
    \sin(\beta) =p_{1,1}\sigma_{k+1} /\sqrt{p_{1,1}^2\sigma_{k+1}^2+p_{2,1}^2\sigma_k^2}. 
    \end{cases}
    \label{opt:sol-rab}
\end{align}
Plug the optimal $r$, $\alpha$, $\beta$ of~\eqref{opt:sol-rab} into $f(\mathbf{v},\lambda)$, and we have
\begin{align*}\nonumber
  f(\lambda) &\triangleq \min_{v_1^2+v_3^2=v_2^2+v_4^2}f(\mathbf{v},\lambda) \\\nonumber
  & = (\sigma_k^2+\sigma_{k+1}^2)/2\\ 
  &\quad- \sqrt{p_{1,1}^2\sigma_k^2+p_{2,1}^2\sigma_{k+1}^2}\sqrt{p_{2,1}^2\sigma_k^2+p_{1,1}^2\sigma_{k+1}^2}.
\end{align*}
According to inequality~\eqref{ineq:feasibility}, inequality $f(\lambda) \le \eta^2$ now is equivalent to
\begin{align}\nonumber
    \sqrt{p_{1,1}^2\sigma_k^2+p_{2,1}^2\sigma_{k+1}^2}\sqrt{p_{2,1}^2\sigma_k^2+p_{1,1}^2\sigma_{k+1}^2} \\\ge (\sigma_k^2+\sigma_{k+1}^2)/2 - \eta^2. 
    \label{ineq:t1}
\end{align}
Denote the right hand of the above inequality as $c\triangleq (\sigma_k^2+\sigma_{k+1}^2)/2 - \eta^2$. Since $\eta < (\sigma_k - \sigma_{k+1})/\sqrt{2}$, we have $c > \sigma_k\sigma_{k+1}$. Furthermore, we notice that $p_{1,1}^2 = 1- p_{2,1}^2$. Plug it into  inequality~\eqref{ineq:t1}, and we have
\begin{align}
    p_{2,1}^4 - p_{2,1}^2 + \frac{c^2 - \sigma_k^2\sigma_{k+1}^2}{(\sigma_k^2-\sigma_{k+1}^2)^2} \le 0.
    \label{ineq:p}
\end{align}
Let 
\begin{align}
    w \triangleq \frac{c^2 - \sigma_k^2\sigma_{k+1}^2}{(\sigma_k^2-\sigma_{k+1}^2)^2},
    \label{eq:no-rank-w}
\end{align} 
and since $\sigma_k\sigma_{k+1} < c \le (\sigma_k^2+\sigma_{k+1}^2)/2$, we have $0 < w \le \frac{(\sigma_k^2 + \sigma_{k+1}^2)^2/4 - \sigma_k^2\sigma_{k+1}^2}{(\sigma_k^2-\sigma_{k+1}^2)^2} = 1/4$. Denote the left hand of inequality~\eqref{ineq:p} as $h(p_{2,1})$, and we have 
\begin{align*}
    h_{\min} &= h(1/\sqrt{2}) = -1/4 + w \le 0, \\
    h(1) &= w >0. 
\end{align*}
Moreover, since $1/\sqrt{2}<  p_{2,1} < 1$, we must have 
\begin{align}
    p_{2,1} \le p_{2,1}^H, 
    \label{ineq:p21}
\end{align}
where $p_{2,1}^H = \sqrt{(1+\sqrt{1-4w})/2}$ is the largest root of $h(p_{2,1}) =0$. Pluging the expressions of $p_{2,1}$ and $p_{2,1}^H$ into \eqref{ineq:p21}, we can get 
\begin{align*}
    \frac{\sqrt{\lambda^2+1}+\lambda}{\sqrt{(\sqrt{\lambda^2+1}+\lambda)^2+1}}\le\sqrt{\frac{1+\sqrt{1-4w}}{2}}. 
\end{align*}
Simplifying this inequality leads to $\lambda \le \frac{e^2 -1 }{2e}$, where 
\begin{align}
e = \sqrt{\frac{1+\sqrt{1-4w}}{1-\sqrt{1-4w}}}.
\label{eq:no-rank-e}
\end{align}
Thus we can conclude that 
\begin{align}
    \lambda_{\max} = \frac{e^2 -1 }{2e}.
    \label{opt:opt-lambda}
\end{align}
Accordingly, the optimal subspace distance in \eqref{opt:generaloriginal} is 
\begin{align}
    \theta^* =\text{atan}(\lambda_{\max} )/2.
    \label{opt:no-rank-theta}
\end{align}
In summary, given energy budget $\eta$, we first compute $w$ according to~\eqref{eq:no-rank-w} and compute $e$ according to~\eqref{eq:no-rank-e}, from which we can get $\lambda_{\max}$ and $\theta^*$ using~\eqref{opt:opt-lambda} and \eqref{opt:no-rank-theta}. Having obtained the optimal $\lambda_{\max}$, we can compute $\mathbf{P}$ in \eqref{eq:no-rank-P} and compute $\mathbf{v}$ using~\eqref{opt:sol-rab} and \eqref{eq:no-rank-x}, and sequentially compute $\mathbf{u}$ and $\bar{\mathbf{b}}$. Finally, if the optimal solution of problem~\eqref{opt:no-rank-simplified} is $\mathbf{B}^*$ with non-zero entries $\bar{\mathbf{b}}^* = [b_{k,k}^*,b_{k+1,k}^*,b_{k,k+1}^*,b_{k+1,k+1}^*]^\top$, we also have another paired feasible optimal solution with non-zero entries being $[b_{k,k}^*,-b_{k+1,k}^*,-b_{k,k+1}^*,b_{k+1,k+1}^*]^\top$,
which leads to the same optimal value. Accordingly, the optimal solution to problem~\eqref{opt:generaloriginal} is $\mathbf{\Delta X}^* = \mathbf{U}\mathbf{B}^*\mathbf{V}^\top$. 



\section{Numerical experiments and applications}\label{sec:num}
In this section, we provide numerical examples to illustrate the results obtained in this paper. We will also apply the results to principal component regression\cite{jackson2005user} to illustrate potential applications in practice. 

\subsection{Numerical experiments}
In this subsection, we illustrate the results with synthesized data.

In the first experiment, we employ different attack strategies in a low-rank data matrix. In this simulation, we set $d=5$, $n=5$, and $k=3$. We generate the original data matrix as $\mathbf{X} = \mathbf{A}\mathbf{B}^\top$, where $\mathbf{A}\in \mathbb{R}^{d \times k}$, $\mathbf{B}\in \mathbb{R}^{n \times k}$, and each entry of $\mathbf{A}$ and $\mathbf{B}$ is i.i.d. generated according to a standard normal distribution. 
First, we conduct our optimal rank-one attack strategy. In this strategy, we use the result from the analysis of the optimal rank-one modification matrix to design $\mathbf a,\mathbf b$ and add the attack matrix $\mathbf{\Delta  X}=\mathbf a\mathbf b^\top$ to the original data matrix $\mathbf X$. We then perform SVD on $\hat{\mathbf X}$ and select the $k$ leading principal components. Finally, we compute the distance between the selected subspace and the original subspace. 
We also conduct a test using a random rank-one attack strategy, in which we randomly generate $\mathbf a, \mathbf b$ with each entry of $\mathbf a, \mathbf b$ being i.i.d. generated according to the standard normal distribution. Then we normalize the energy of $\mathbf a\mathbf b^\top $to be $\eta^2$. For each $\eta$, we repeatedly generate 100000 pairs of $\mathbf a$ and $\mathbf b$ and compute their corresponding subspace distances. 
In addition, we compare it with the strategy where the modification matrix is free of rank constraint. Although our analysis is deliberately designed for general data matrices, we set the ($k+1$)th singular value to be zero so that it can be applied to the low-rank data matrix. We design the modification matrix $\mathbf{\Delta X}$ according to our analysis in this paper and calculate the subspace distance between the original subspace and that after modification. 
Moreover, we conduct another random attack strategy in which we randomly generate the modification matrix without any rank constraint. Each entry of the modification matrix is i.i.d. generated according to a standard normal distribution. After that, we normalize its Frobenious norm equal to $\eta$. We repeat this attack 100000 times for each $\eta$ and record its corresponding subspace distance.  
Furthermore, we also compare it with the strategy described in \cite{pimentel2017adversarial}, which adds one adversarial data sample into the data set. 

\begin{figure}[t!]
    \centering
    \includegraphics[width=.68\linewidth]{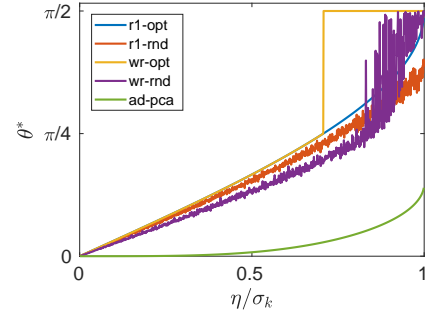}
    \caption{Subspace distances with different attack strategies on a low-rank data matrix over different energy budgets.}
    \label{fig:low-rank}
\end{figure}

Fig.~\ref{fig:low-rank} demonstrates the subspace distances obtained by the five strategies. In this figure, r1-opt represents the rank-one optimal attack obtained in this paper, r1-rnd represents the maximal subspace distance obtained among the 100000 times random rank-one attacks, wr-opt stands for our optimal attack without the rank constraint, wr-rnd is the maximal subspace distance among the 100000 random attacks without the rank constraint, and ad-pca is the algorithm described in \cite{pimentel2017adversarial}. The $x$ axis is the ratio between $\eta$ and the smallest singular value of the original data matrix. From the figure, we  can see our optimal strategies are much better than the ad-pca strategy. It is because our strategies can modify the data matrix, and thus have higher degrees of freedom to manipulate the data. The optimal strategies designed in this paper also have a larger subspace distance compared with their corresponding random attack strategies. In the region where $\eta/\sigma_{k}\in[0,1/\sqrt{2}]$, both of our two optimal strategies provide the same subspace distances, 
which can be verified by setting $\sigma_{k+1}=0$, computing $\theta^*$ in equation~\eqref{opt:no-rank-theta} and comparing it with the value in  equation~\eqref{opt:rank-one-theta}. 
When $\eta/\sigma_k > 1/\sqrt{2}$, the optimal attack without the rank constraint leads to the largest subspace distance, $\pi/2$, which is much larger than the distance obtained by the optimal rank-one attack strategy. That means, without the rank constraint, it indeed provides a larger subspace distance.

In the second numerical experiment, we test these strategies except the ad-pca in the general data matrix in which the data matrix is not low-rank. In this experiment, we set $d=5$, $n=5$, and $k=3$. We randomly generate the data matrix $\mathbf{X}\in\mathbb{R}^{d\times n}$ with each entries i.i.d generated according to a standard normal distribution. We also design the optimal rank-one attack matrix and the optimal modification matrix without the rank constraint according to the analysis provided in this paper. In addition, we do random attacks 100000 times using the randomly generated modification matrix with the rank-one constraint and without the rank constraint respectively.   

\begin{figure}[t!]
    \centering
    \includegraphics[width=.68\linewidth]{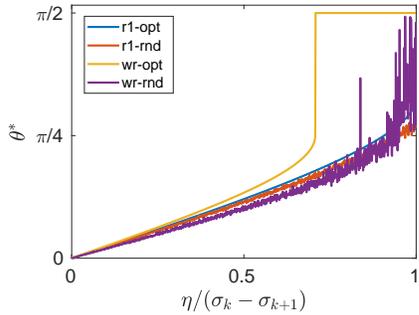}
    \caption{Subspace distances achieved by using different attack strategies under different energy budgets.}
    \label{fig:general}
\end{figure}
Fig. \ref{fig:general} shows the subspace distances obtained through different strategies over different energy budgets. In this figure, the $x$ axis is the ratio between $\eta$ and $\sigma_k - \sigma_{k+1}$. For the two random attack strategies, we demonstrate the maximal subspace distances achieved by the 100000 times random attacks. As the figure shows, both of the two random strategies have smaller subspace distances compared with their perspective optimal strategies. Different from the low-rank case, the strategy without the rank constraint provides larger subspace distances consistently over all the energy budgets.

\subsection{Applications}
In this subsection, we use real data to illustrate the results obtained in this paper.

In particular, we illustrate the impact of adversarial attack on PCR, which is widely used in statistical learning especially when collinearity exists in the data. Ordinary regression will increase the standard error of the coefficients when there are high correlations or even collinearities between features. This happens particularly when the number of features is much larger than the number of data samples. PCR deals with this issue by performing PCA on the feature matrix and only selecting the leading $k$ principal components as the predictors, and thus dramatically decreases the number of predictors.
The regression process of PCR can be seen as projecting the response values onto the subspace spanned by the leading $k$ principal components. So, the accuracy of the subspace will significantly influence the regression results. 
More details of PCR can be found in~\cite{jackson2005user}.  
\begin{figure}
    \centering
    \includegraphics[width=0.69\linewidth]{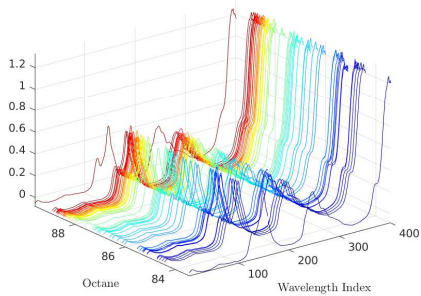}
    \caption{The spectral intensities of the gasaline data set.}
    \label{fig:app-data-description}
\end{figure}

In this experiment, our task is to use the gasoline spectral intensity to predict its octane rating. We use the gasoline spectral data set~\cite{kalivas1997two}, which comprises spectral intensities of 60 samples of gasoline at 401 wavelengths, and their octane ratings. 
Fig.~\ref{fig:app-data-description} shows the spectral intensities of the data set. This figure indicates that the correlation of intensity among different wavelengths is very high. To complete the regression task, we can use PCR. 

In this experiment, we randomly select $80$ percent of the data as the training set and the remaining $20$ percent as the test set. We choose $4$ principal components as our predictors and perform regression based on these principal components. We also record the r-squared values both in the training phase and the test phase. 
The r-squared value is defined as $r^2 = 1- \frac{\|\mathbf{y}-\hat{\mathbf{y}}\|^2}{\|\mathbf{y}-\bar{\mathbf{y}} \|^2}$, where $r^2$ is the r-squared value, $\mathbf{y}$ is the response values, $\hat{\mathbf{y}}$ is the predicted values, $\|\mathbf{y}-\bar{\mathbf{y}}\|^2$ represents the total variance of the response values, and $\bar{\mathbf{y}}=\textrm{mean}(\mathbf{y})\cdot\mathbf{1}$ stands for the mean vector of the response values. R-squared value measures how well the model fits the data and larger r-squared value indicates better regression.
Firstly, we perform regular PCR without attack and let na-train and na-test denote the r-squared values of the training and test respectively. We then attack the feature matrix using the optimal rank-one strategy proposed in this paper with different energies and denote r1-train and r1-test as its r-squared values in the training and test processes. Finally, we also carry out the optimal attack without the rank constraint and denote wr-train, wr-test as the r-squared values in the training and test procedures. 

\begin{figure}[t!]
    \centering
    \includegraphics[width=.68\linewidth]{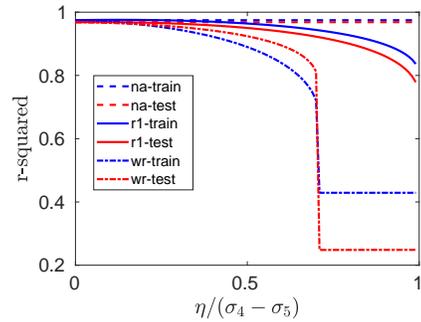}
    \caption{R-squared values with different attack strategies over different energy budgets.} 
    \label{fig:app-pcr-attcks}
\end{figure}

Fig.~\ref{fig:app-pcr-attcks} illustrates the r-squared values with different attack strategies under different energy budgets. As shown in this figure, with the increasing of the energy budget, r-squared values of training and test decrease for both attack strategies. This figure also indicates that the strategy with no rank constraint is more efficient than the rank-one strategy considering its smaller r-squared values. Furthermore, the r-squared value of the strategy without the rank constraint has a tremendous drop at the point $\eta/(\sigma_4-\sigma_5)=1/\sqrt{2}$, which is consistent with our analysis that beyond this particular point the maximal subspace distance is $\pi/2$.

\section{Conclusion}\label{sec:con}
In this paper, we have investigated the adversarial robustness of PCA problem. We have characterized the optimal rank-one adversarial modification strategy and the optimal strategy without the rank constraint to modify the data. Our analysis has showed that both of the two strategies depend on the singular values of the data matrix and the adversary's energy budget. We have also performed numerical simulations and investigated the impact of this attack on PCR. Both the numerical experiments and the PCR application illustrate that adversarial attacks degrade the performance of subspace learning significantly. In the future, it is of interest to investigate the defense strategy to mitigate the effects of this attack.


\appendices
\section{Poof of the equivalence between problem~\eqref{opt:low_rank_1} and problem~\eqref{opt:low_rank_3}} \label{app:eq12-eq13}
Before giving the proof, we first examine the unitary invariant property of the Asimov distance, which is helpful in our subsequent proof.
\begin{proposition}	
Let $\mathbf P$ and $\mathbf T$ be unitary matrices,  and then for the Asimov distance function $\theta(\cdot,\cdot)$, we have 
$$\theta\big(\mathbf X_1,g_k(\mathbf X_2)\big) 
= \theta\big(\mathbf P\mathbf X_1 \mathbf T^\top ,g_k(\mathbf P\mathbf X_2 \mathbf T^\top )\big).$$
\label{prop:p1}
\end{proposition}

\begin{proof}
First, we show $\theta(\mathbf X_1,\mathbf X_2) = \theta(\mathbf P\mathbf X_1\mathbf T^\top,\mathbf P\mathbf X_2\mathbf T^\top)$. Suppose the thin QR decompositions of $\mathbf X_1$ and $\mathbf X_2$ are 
$\mathbf X_1 = \mathbf Q_1 \mathbf R_1,\quad \mathbf X_2 = \mathbf Q_2\mathbf R_2$, and then the subspace distance between the two subspaces spanned by the columns of $\mathbf X_1$ and $\mathbf X_2$ is determined by the singular values of $\mathbf Q_1^\top \mathbf Q_2$. Since $(\mathbf P \mathbf Q_1)^\top (\mathbf P \mathbf Q_2) = \mathbf Q_1^\top \mathbf Q_2$ and right multiplying an unitary matrix does not change the singular values and the column subspace of a matrix, we have $\theta(\mathbf X_1,\mathbf X_2) = \theta(\mathbf P\mathbf X_1\mathbf T^\top,\mathbf P\mathbf X_2\mathbf T^\top)$. 

Second, suppose the full SVD of $\mathbf X_2$ is $\mathbf X_2 = \mathbf U_2\mathbf \Sigma_2 \mathbf V_2^\top $, where $\mathbf U_2 = [\mathbf u_{21},\mathbf u_{22},\cdots, \mathbf u_{2d}]$. Then 
$$\mathbf P g_k(\mathbf X_2) = \mathbf P[\mathbf u_{21},\mathbf u_{22},\cdots,\mathbf u_{2k}] = g_k(\mathbf P\mathbf X_2),$$
which can be verified by checking that $\mathbf P\mathbf U_2\mathbf \Sigma_2 \mathbf V_2^\top $ is a valid SVD of $\mathbf P\mathbf X_2$. It completes the proof. 
\end{proof}
With the help of this proposition, let $\mathbf P = \mathbf U^\top $, $\mathbf T=\mathbf V^\top$, right multiply $\mathbf{P}$ and left multiply $\mathbf{T}^\top$ on both $\mathbf{X}$ and $\mathbf{\hat{X}}$, and we can simplify problem~\eqref{opt:low_rank_1} as the following
\begin{align}
    \max_{\mathbf a\in \mathbb R^{d},\mathbf b\in \mathbb R^{n}}:\quad & \theta(\mathbf \Sigma, g_k(\tilde{\mathbf Y}))     \label{opt:low_rank_2}\\\nonumber 
    \text{s.t.}\quad & \tilde{\mathbf Y} = \mathbf \Sigma + \mathbf a\mathbf b^\top, \\ \nonumber
    &\|\mathbf a\|\|\mathbf b\| \le \eta,  
\end{align}
where we assume $n>d$, $\mathbf{\Sigma}=[\textrm{diag}(\sigma_1,\sigma_2,\cdots,\sigma_k,\mathbf{0}),\mathbf{0}]\in \mathbb{R}^{d\times n}$. 
Also, from problem~\eqref{opt:low_rank_1} to problem~\eqref{opt:low_rank_2}, we do variable change $\mathbf a \Leftarrow \mathbf{U}^\top \mathbf a, \mathbf b \Leftarrow \mathbf{V}^\top \mathbf{b}$. 

To further simplify this optimization problem, we split $\mathbf a$ and $\mathbf b$ into $\mathbf a=[\mathbf a_1^\top,\mathbf a_2^\top]^\top, \mathbf b=[\mathbf b_1^\top, \mathbf b_2^\top]^\top$, where $\mathbf a_1\in \mathbb{R}^{k}$, $\mathbf a_2 \in \mathbb{R}^{d-k}$, $\mathbf b_1\in \mathbb{R}^{k}$, and $\mathbf b_2\in \mathbb{R}^{n-k}$. In addition, utilizing the Householder transformation\cite{horn2012matrix}, we construct an orthogonal matrix 
\begin{equation}
\mathbf M_1 = 
\begin{bmatrix}
\mathbf I_k & \mathbf 0 \\
\mathbf 0 & \mathbf H_1
\end{bmatrix},
\end{equation}
where 
\begin{align*}
    &\mathbf M_1^\top \mathbf M_1 = \mathbf I,\,
    &&\mathbf H_1 = \mathbf I - 2\frac{\mathbf u\mathbf u^\top}{\|\mathbf u\|^2},\\
    &\mathbf u = \mathbf a_2 - s_1\|\mathbf a_2\| \cdot \mathbf e_1,
    &&\mathbf e_1 =[1,0,\cdots,0]^\top \in \mathbb{R}^{d-k},\\
    & \mathbf{H}_1^\top \mathbf{a}_2 = s_1\|\mathbf{a}_2\|\cdot\mathbf{e}_1, 
    && s_1 = \pm 1.
\end{align*}
Similarly, we can construct another Householder transformation matrix $\mathbf{H}_2$ for $\mathbf{b}_2$ and the corresponding orthogonal matrix $\mathbf{M}_2 = \text{diag}(\mathbf{I}_k, \mathbf{H}_2)$. Left multiplying $\mathbf{M}_1^\top$ and right multiplying $\mathbf{M}_2$ on $\tilde{\mathbf{Y}}$, we have 
\begin{align*}
    \mathbf{M}_1^\top \tilde{\mathbf{Y}}\mathbf{M}_2 = \begin{bmatrix}
    \tilde{\mathbf{\Sigma}} & \mathbf{0} \\
    \mathbf{0}     & \mathbf{0}
    \end{bmatrix} 
    + \begin{bmatrix}
    \mathbf{a}_1 \\
    s_1 \|\mathbf{a}_2\| \\
    \mathbf{0}
    \end{bmatrix}
    \begin{bmatrix}
    \mathbf{b}_1^\top & s_2\|\mathbf{b}_2\| &  \mathbf{0}
    \end{bmatrix},
\end{align*}
where $s_2 = \pm 1$. 
 
Let $\mathbf{a}\Leftarrow [\mathbf{a}_1^\top,\,  s_1\|\mathbf{a}_2\|]^\top$ and $\mathbf{b}\Leftarrow [\mathbf{b}_1^\top,\, s_2\|\mathbf{b}_2\|]^\top$. Utilizing Proposition \ref{prop:p1}, it is clear that problem~\eqref{opt:low_rank_3} and problem~\eqref{opt:low_rank_2} are equivalent. 

\section{Proof of Theorem~\ref{thm:ab}}\label{app:thm1}
The proof follows similar steps to those in~\cite{pimentel2017adversarial}. In problem~\eqref{opt:low_rank_3}, $\tilde{\mathbf \Sigma}$ is a diagonal matrix with diagonal elements $\{\sigma_1,\sigma_2,\cdots,\sigma_k,0\}$. The subspace spanned by $g_k(\mathbf Y)$ is a $k$ dimensional subspace in $\mathbb R^{k+1}$. We denote this subspace as $\mathbb{Q}$, denote $\mathbb{P}$ as the subspace spanned by $\tilde{\mathbf \Sigma}$ and further denote their intersection as $\mathbb{T} = \mathbb{P} \cap \mathbb{Q}$.
Note that $\mathbb P$ is not equal to $\mathbb Q$ (otherwise the Asimov distance will be zero), so we have $\textrm{dim}(\mathbb P \cup \mathbb Q)=k+1$.
Since $\textrm{dim}(\mathbb{P})+\textrm{dim}(\mathbb{Q})-\textrm{dim}(\mathbb{T}) =\textrm{dim}(\mathbb P \cup \mathbb Q), $
we have $\textrm{dim}(\mathbb{T}) = k-1$. Let $\mathbf T$ be an orthonormal basis of $\mathbb T$. Let $[\mathbf T, \mathbf p]$ be an orthonormal basis of $\mathbb{P}$ and let $[\mathbf T,\mathbf q]$ be an orthonormal basis of $\mathbb{Q}$. By the definition of Asimov distance, the subspace distance between $\mathbb{P}$ and $\mathbb{Q}$ is the angle between $\mathbf p$ and $\mathbf q$. 

Firstly, it is easy to see that $a_{k+1}\neq 0$. Otherwise, $\mathbb{Q}$ will be equal to $\mathbb{P}$, which means that their Asimov distance is zero. 

Secondly, it is easy to see $\mathbf q \in \textrm{span}[\mathbf T, \mathbf p,\mathbf e_{k+1}]$, where $\mathbf e_{k+1}$ is an ordinary basis vector that only has element $1$ in the ($k+1$)th coordinate.  Since $\mathbf T$ is orthogonal to $\mathbf q$, we have $\mathbf q \in \textrm{span}[\mathbf p, \mathbf e_{k+1}]$. It is easy to see that the larger variance in the direction of $\mathbf p$ is, the closer $\mathbf p$ and $\mathbf q$ will be. Then we should select $\mathbf p$ as the direction with the smallest variance in $\mathbf X$. Since we are assuming that $\sigma_1 \ge \sigma_2 \ge \cdots \ge \sigma_k$, $\mathbf p$ should be $\mathbf e_{k}$. 

Thirdly, for a fixed direction of $\mathbf a$, let $\hat{\mathbf a}$ be the projection of $\mathbf a$ onto $\text{span}[\mathbf e_k,\mathbf e_{k+1}]$. Clearly, $\mathbf q$ will be closer to $\mathbf{a}$ as $\hat{\mathbf a}$ grows. As a result, the angle between $\mathbf q$ and $\mathbf p$ will be larger. This also implies that the length of $\mathbf a$ should be maximized: $\|\mathbf a\| = \eta$. Hence, the Asimov distance is maximized when $\mathbf a = \hat{\mathbf{a}}$ and $\|\mathbf a\|=\eta$, implying that $\mathbf a$ only has nonzero elements in its $k$th and $k+1$th coordinates. 

Finally, for a fixed $\mathbf a$ in the form of \eqref{opt:ab}, the projected variance of $\mathbf Y$ on the direction of $\mathbf e_k$ is $v_1 = \sum_{i\neq k} (a_k b_i)^2  + (a_k b_k+\sigma_k)^2 = a_k^2+\sigma_k^2+2a_k b_k\sigma_k$ and the projected variance of $\mathbf Y$ on the direction of $\mathbf e_{k+1}$ is $v_2 = \sum_i (a_{k+1}b_i)^2 = a_{k+1}^2$. To maximize the Asimov distance, we need to make $v_1$ small and $v_2$ large. Apparently, for fixed $\mathbf a$, $v_1$ is minimized when $b_k = -\text{sign}(a_k)$, which implies $b_i=0, \forall i\neq k$. To avoid the sign ambiguity, we set $b_k=1$. 

\section{Proof of Theorem~\ref{theorem:rank-one-optimality}}\label{app:thm2}
This proof follows the similar steps in the proof of the low-rank case. 
Denote $\mathbb{P}$ as the subspace spanned by $g_k(\mathbf \Sigma)$ and $\mathbb{Q}$ as the subspace spanned by $g_k(\mathbf{Y})$, and denote their intersection as $\mathbb{T} = \mathbb{P}\cap \mathbb{Q}$. 
We further denote $\mathbf{T}$ as an orthonormal basis of $\mathbb{T}$, $[\mathbf{T},\mathbf{p}]$ as an orthonormal basis of $\mathbb{P}$, and $[\mathbf{T},\mathbf{q}]$ as an orthonormal basis of $\mathbb{Q}$.
From the definition of Asimov distance, the subspace distance between $\mathbb{P}$ and $\mathbb{Q}$ is the subspace distance between the span of $\mathbf{p}$ and the span of $\mathbf{q}$. 

First, it is apparent that $\mathbf{q}\in \text{span}[\mathbf{T},\mathbf{p},\mathbf{e}_{k+1},\mathbf{e}_{k+2},\cdots,\mathbf{e}_{d}]$. Since $\mathbf{q}\perp\mathbf{T}$, we have $\mathbf{q}\in \text{span}[\mathbf p,\mathbf e]$, where $\mathbf e \in  \text{span}[\mathbf{e}_{k+1},\cdots,\mathbf{e}_d]$. It is easy to see that the subspace distance between the span of $\mathbf q$ and the span of $\mathbf p$ will be large if the variance of $\mathbf{\Sigma}$ in the span of $\mathbf{p}$ is large and the variance of $\mathbf{\Sigma}$ in the span of $\mathbf{q}$ is small. So we should select $\mathbf p$ as the direction in $\text{span}[\mathbf{e}_1,\cdots,\mathbf{e}_k]$ that has the smallest variance of $\mathbf{\Sigma}$ and select $\mathbf{e}$ as the direction among $\text{span}[\mathbf{e}_{k+1},\cdots,\mathbf{e}_n]$ that has the largest variance of $\mathbf{\Sigma}$. Since $\mathbf{e}\in \text{span}[\mathbf{e}_1,\mathbf{e}_2,\cdots,\mathbf{e}_d]$ and $\sigma_1 \ge \sigma_2\ge\cdots\ge \sigma_k\ge\sigma_{k+1}\ge\cdots\ge\sigma_d$, $\mathbf{p}$ should be $\mathbf{e}_k$ and $\mathbf{e}$ should be $\mathbf{e}_{k+1}$. So, we have $\mathbf{q} \in \text{span}[\mathbf{e}_k,\mathbf{e}_{k+1}]$. 

Second, for a fixed direction of $\mathbf a$, let $\Hat{\mathbf{a}}$ be the projection of $\mathbf a$ onto $\text{span}[\mathbf{e}_k,\mathbf{e}_{k+1}]$. It is easy to see that $\mathbf{q}$ will be closer to $\mathbf{a}$ as $\Hat{\mathbf{a}}$ grows, and as a result, the angle between $\mathbf{q}$ and $\mathbf{p}$ will be larger. This implies the length of $\mathbf{a}$ should be maximized, which indicates $\|\mathbf{a}\|=\eta$ and the distance is maximized when $\mathbf{a}=\Hat{\mathbf{a}}$. It also indicates $a_i=0$ if $i\neq k,\,k+1$.  

Finally, for a fixed $\mathbf a$ in the form of~\eqref{opt:a}, the projected variance of $\mathbf{Y}$ in the direction of $\mathbf{e}_k$ is 
$v_k = \sum_{i\neq k}(a_kb_i)^2+(a_kb_k+\sigma_k)^2 = a_k^2+\sigma_k^2+2a_kb_k\sigma_k$ and the projected variance of $\mathbf{Y}$ in the direction of $\mathbf{e}_{k+1}$ is $v_{k+1} = \sum_{i\neq k+1}(a_{k+1}b_i)^2+(\sigma_{k+1}+a_{k+1}b_{k+1})^2 = a_{k+1}^2+\sigma_{k+1}^2+2a_{k+1}b_{k+1}$. To maximize the Asimov distance, we should make $v_k$ small and make $v_{k+1}$ large. With the constraint that $\|\mathbf{b}\|=1$, we should have $b_k^2+b_{k+1}^2 =1$, which implies $b_i=0$ for all $i\neq k$ and $i\neq (k+1)$. 

As shown above, the optimal $\mathbf a$ and $\mathbf{b}$ should be in the form of~\eqref{opt:a} and~\eqref{opt:b}, which completes our proof.

\section{Proof of Theorem~\ref{thm:sol-k<r}}\label{app:sol-kkt1}
The optimal solution to problem~\eqref{opt:theta2} either locates at the boundary, or at the stationary points.

We first characterize the stationary points. At the stationary points, the value $(\alpha^*,\beta^*)$ satisfies the necessary conditions
\begin{align}
    \begin{cases}
    &\frac{\partial}{\partial \alpha}|\cos \varphi(\alpha,\beta)|_{\alpha=\alpha^*,\beta=\beta^*} 
    =0, \\
    &\frac{\partial}{\partial \beta}| \cos \varphi(\alpha,\beta)|_{\alpha=\alpha^*,\beta=\beta^*} 
    =0.
    \end{cases}
    \label{eq:kkt0}
\end{align}
Since $\sin \varphi^* \neq 0$, we have
\begin{align}
    \begin{cases}
    \frac{\partial}{\partial \alpha}\varphi(\alpha,\beta)|_{\alpha=\alpha^*,\beta=\beta^*} 
    &=0, \\
    \frac{\partial}{\partial \beta}\varphi(\alpha,\beta)|_{\alpha=\alpha^*,\beta=\beta^*} 
    &=0,
    \end{cases}
    \label{eq:kkt1}
\end{align}
in which
\begin{align*}
\frac{\partial \varphi}{\partial \alpha}  =& \frac{\eta}{a_x^2+a_y^2}
\Big( \eta(3\sigma_{k+1}^2-\sigma_k^2+\eta^2)\\
&+2\eta(\sigma_k^2-\sigma_{k+1}^2)\cos^2(\alpha)+2\eta(\sigma_k^2-\sigma_{k+1}^2)\cos^2(\beta) \\
&+ \sigma_k(\sigma_k^2-\sigma_{k+1}^2+3\eta^2)\cos(\alpha)\cos(\beta) \\
&+ \sigma_{k+1}(\sigma_{k+1}^2-\sigma_k^2+3\eta^2)\sin(\alpha)\sin(\beta) \Big), \\
\frac{\partial \varphi}{\partial \beta}=& \frac{\eta}{a_x^2+a_y^2}
\Big(
\sigma_k(\sigma_{k+1}^2+\eta^2-\sigma_k^2)\sin(\alpha)\sin(\beta)\\
&+\sigma_{k+1}(\sigma_k^2+\eta^2-\sigma_{k+1}^2)\cos(\alpha)\cos(\beta)+2\eta\sigma_k\sigma_{k+1}
\Big).
\end{align*}

Eliminating $\sin(\alpha)\sin(\beta)$ from~\eqref{eq:kkt1}, we have
\begin{align}
C\cos^2(\alpha)+D\cos(\alpha)\cos(\beta)+C\cos^2(\beta)+F = 0,
\label{eq:CDE}
\end{align}
where 
$
C= 2\eta\sigma_k(\sigma_k^2-\sigma_{k+1}^2)(\sigma_{k+1}^2+\eta^2-\sigma_k^2 )$, $
D= (\sigma_k^2-\sigma_{k+1}^2)\left(-(\sigma_k^2-\sigma_{k+1}^2)^2-2\eta^2(\sigma_k^2+\sigma_{k+1}^2)+3\eta^4\right) $ and
$F= \eta\sigma_k\left(\sigma_k^4+\sigma_{k+1}^4+\eta^4-2\sigma_k^2\sigma_{k+1}^2-2\sigma_k^2\eta^2-2\sigma_{k+1}^2\eta^2\right).
$

Further, we rewrite the first equation of \eqref{eq:kkt1} as
\begin{align}
    c\sqrt{(1-\cos^2(\alpha))(1-\cos^2(\beta))}+d\cos(\alpha)\cos(\beta)+e =0,
    \label{eq:cde}
\end{align}
where 
$ c = \sigma_k(\sigma_{k+1}^2+\eta^2-\sigma_k^2)$, $
    d =\sigma_{k+1}(\sigma_k^2+\eta^2-\sigma_{k+1}^2),$ and $
    e =2\eta \sigma_k\sigma_{k+1}.$
    
Combining \eqref{eq:CDE} and \eqref{eq:cde} and eliminating $\cos^2(\alpha)$ and $\cos^2(\beta)$, we have
\begin{align*}
    &(c^2-d^2)\cos(\alpha)^2\cos(\beta)^2\\
    &+\left(\frac{Dc^2}{C}-2de\right)\cos(\alpha)\cos(\beta)+\frac{c^2F}{C}+c^2-e^2=0.
\end{align*}
The left side of the equation is a quadratic function with respect to $r=\cos(\alpha)\cos(\beta)$. The two roots are:
\begin{align*}
    r_1 = -\frac{\sigma_k\eta}{\sigma_k^2-\sigma_{k+1}^2}, 
    r_2 =-\frac{\sigma_k}{2}\left(\frac{1}{\eta}+\frac{\eta}{\sigma_k^2-\sigma_{k+1}^2}\right).
\end{align*}
Note that $\eta \in[0,\sigma_k-\sigma_{k+1})$, so we have $r_1 \in (-\frac{\sigma_k}{\sigma_k+\sigma_{k+1}},0]$, $r_2 \in(-\infty,-\frac{\sigma_k}{\sigma_k-\sigma_{k+1}})$. Since $|\cos(\alpha)\cos(\beta)|\le 1$, $\frac{\sigma_k}{\sigma_k+\sigma_{k+1}}<1$, and $\frac{\sigma_k}{\sigma_k-\sigma_{k+1}}>1$, we should only retain the first root $r_1$. Substitute $\cos(\alpha)\cos(\beta)=r_1=-\frac{\eta\sigma_k}{\sigma_k^2-\sigma_{k+1}^2}$ into~\eqref{eq:CDE}, and we have 
 $   C\cos^4(\alpha)+(Dr_1+F)\cos^2(\alpha)+Cr_1^2=0.$
The left side of the equation is a quadratic function with respect to $s=\cos^2(\alpha)$, so we can easily find its roots. Let us denote $s_1$ and $s_2$ as the two roots: 
\begin{align*}
    s_1 = \frac{\sigma_k^2-\sigma_{k+1}^2+\eta^2-\sqrt{H}}{2(\sigma_k^2-\sigma_{k+1}^2)},     s_2 =\frac{\sigma_k^2-\sigma_{k+1}^2+\eta^2+\sqrt{H}}{2(\sigma_k^2-\sigma_{k+1}^2)},
\end{align*}
where $H = \sigma_k^4+\sigma_{k+1}^4+\eta^4-2\sigma_k^2\sigma_{k+1}^2-2\sigma_k^2\eta^2-2\sigma_{k+1}^2\eta^2$. We need to check that $H$ is positive. Viewing $H$ as a function of $\eta$ and taking derivative, we have
$H^\prime(\eta) =2\eta(2\eta^2-2(\sigma_k^2+\sigma_{k+1}^2))<0$. Since $\eta^2\in[0,(\sigma_k-\sigma_{k+1})^2)$, we have 
$H(\eta) \in (0, (\sigma_k^2-\sigma_{k+1}^2)^2]$.

As $\cos(\alpha)^2\le1$, we need to check whether $s_1,s_2\in[0,1]$.

Firstly, as $H$ is a decreasing function of $\eta$ in the considered range, $s_1$ is a increasing function of $\eta$. Therefore, we have $\min(s_1) = s_1(\eta)|_{\eta=0} = 0$ and $\max(s_1) = s_1(\eta)|_{\eta=\sigma_k-\sigma_{k+1}} = \frac{\sigma_k}{\sigma_k+\sigma_{k+1}}<1$. Hence, $s_1$ is a valid solution. 

Secondly, it is easy to check that $s_2$ is a decreasing function of $\eta$. 
So, we have 
$\max(s_2) = s_2(\eta)|_{\eta=0}=1$ and $\min(s_2)=s_2(\eta)|_{\eta = \sigma_k-\sigma_{k+1}} = \frac{\sigma_k}{\sigma_k+\sigma_{k+1}}<1$, which means $s_2$ is also a valid solution. 
Hence, we have two stationary points 
\begin{align}
    \begin{cases}
    \cos^2(\alpha) &=\frac{\sigma_k^2-\sigma_{k+1}^2+\eta^2 \pm \sqrt{H}}{2(\sigma_k^2-\sigma_{k+1}^2)}, \\
    \cos^2(\beta)  &= \frac{\sigma_k^2-\sigma_{k+1}^2+\eta^2\mp \sqrt{H}}{2(\sigma_k^2-\sigma_{k+1}^2)}.
    \end{cases}\label{opt:cos}
\end{align}



Since there are two sets of solutions in~\eqref{opt:cos}, we should determine which one is better. The variance of $\mathbf{Y}$ in the direction of $\mathbf{e}_k$ is $v_k =\cos^2(\alpha)+\sigma_k^2+2\cos(\alpha)\cos(\beta)$ and the variance of $\mathbf{Y}$ in the direction of $\mathbf{e}_{k+1}$ is $v_{k+1} = \sin^2(\alpha)+\sigma_{k+1}^2+2\sin(\alpha)\sin(\beta)$. Both of the two sets of solutions in~\eqref{opt:cos} lead to $\cos(\alpha)\cos(\beta)=-\frac{\eta\sigma_k}{\sigma_k^2-\sigma_{k+1}^2}$ and $\sin(\alpha)\sin(\beta) = \frac{\eta\sigma_{k+1}}{\sigma_k^2-\sigma_{k+1}^2}$. For fixed $\cos(\alpha)\cos(\beta)$ and $\sin(\alpha)\sin(\beta)$, the smaller $\cos^2(\alpha)$ is, the smaller $v_k$ will be, and the larger the subspace distance will be. Hence, we conclude the stationary point that satisfies
\begin{align}
\begin{cases}
\cos^2(\alpha^*) &=\frac{\sigma_k^2-\sigma_{k+1}^2+\eta^2 - \sqrt{H}}{2(\sigma_k^2-\sigma_{k+1}^2)} \\
\cos^2(\beta^*)  &= \frac{\sigma_k^2-\sigma_{k+1}^2+\eta^2+ \sqrt{H}}{2(\sigma_k^2-\sigma_{k+1}^2)} 
\end{cases}
\label{sol:stat-point}
\end{align}
leads to a larger subspace distance. 

Finally, it is easy to compute the objective values of problem~\eqref{opt:theta2} at the boundary points. Comparing these values with the objective values induced by the point in equation~\eqref{sol:stat-point}, we can readily conclude the point in equation~\eqref{sol:stat-point} gives a larger objective value. In summary, given that $\alpha \in [0,\pi/2]$ and $\beta\in[\pi/2,\pi]$, the optimal $\alpha$ and $\beta$ are shown in~\eqref{opt:albe}. 

\section{Proof of theorem~\ref{theorem:no-rank-optimality}}\label{app:thm3}
The proof has two main steps. In the first step, we show that non-zero entries of $\mathbf{B}$ are in the $k$th and ($k+1$)th rows. In the second step, we will further prove the entries except in the $k$th and ($k+1$)th columns should be zero. 

In the first step, we follow similar proof procedures in Theorem~\ref{theorem:rank-one-optimality}. We use $\mathbb{P}$ to denote the subspace spanned by $g_k(\mathbf{\Sigma})$ and $\mathbb{Q}$ to denote the subspace spanned by $g_k(\mathbf{Y})$. We also use $\mathbb{T}$ to represent the intersection of the two subspaces and further denote $\mathbf{T}$ as one set of orthonormal bases of $\mathbb{T}$, $[\mathbf{T},\mathbf{p}]$ as one set of orthonormal bases of $\mathbb{P}$ and $[\mathbf{T},\mathbf{q}]$ as one set of orthonormal bases of $\mathbb{Q}$. So, the subspace distance between $\mathbb{P}$ and $\mathbb{Q}$ is the subspace distance between the subspace spanned by $\mathbf{p}$ and that spanned by $\mathbf{q}$. Following the same arguments in Theorem~\ref{theorem:rank-one-optimality}, by setting all the entries of $\mathbf{B}$ to be zero except the $k$th and ($k+1$)th rows, we can guarantee achieving the maximal subspace distance and further we have $\mathbf{q} \in \text{span}[\mathbf{e}_k, \mathbf{e}_{k+1}]$ and $\mathbf{p}=\mathbf{e}_k$. 

In the second step, since the non-zero elements of $\mathbf{B}$ only locate in the $k$th and ($k+1$)th rows and $\mathbf{q}\in\text{span}[\mathbf{e}_k,\mathbf{e}_{k+1}]$, it indicates $\mathbf{q}$ is the direction with the maximal variance on the span of $\mathbf{e}_k$ and $\mathbf{e}_{k+1}$. Assuming $\mathbf{q}=[0,\cdots,\cos(\gamma),\sin(\gamma),\cdots,0]^\top$ with $\cos(\gamma)$ and $\sin(\gamma)$ being in the $k$th and ($k+1$)th coordinates respectively and according to the definition of principal components, we can find $\gamma$ by solving the optimization problem
\begin{align}
    \underset{\gamma}{\mathrm{argmax}}:\quad \mathbf{q}^\top \mathbf{Y}\mathbf{Y}^\top \mathbf{q}. 
\end{align}
Plug $\mathbf{q}=[0,\cdots,\cos(\gamma),\sin(\gamma),\cdots,0]^\top$ into the objective function, and we have 
\begin{align}
    &\quad \mathbf{q}^\top \, \mathbf{Y}\mathbf{Y}^\top \, \mathbf{q}
    =&
    \begin{bmatrix}
    \cos(\gamma)\\
    \sin(\gamma) 
    \end{bmatrix}^\top 
    \begin{bmatrix}
    b_{x1} & 
    \frac{1}{2}b_y\\
    \frac{1}{2}b_y
    & b_{x2}
    \end{bmatrix}
    \begin{bmatrix}
    \cos(\gamma)\\
    \sin(\gamma) 
    \end{bmatrix},
    \label{opt:obj-complex}
\end{align}
where 
$    b_{x1} =  \|\mathbf{b}_k+\mathbf{e}_k\sigma_k\|^2$, $    b_{x2} =  \|\mathbf{b}_{k+1}+\mathbf{e}_{k+1}\sigma_{k+1}\|^2$, $
     b_y   = 2(\mathbf{b}_k+\mathbf{e}_k\sigma_k)^\top (\mathbf{b}_{k+1}+\mathbf{e}_{k+1}\sigma_{k+1}),
$
with $\mathbf{b}_k$ and $\mathbf{b}_{k+1}$ being the transpose of the $k$th and ($k+1$)th rows of $\mathbf{B}$ respectively and $\mathbf{e}_k \in \mathbb{R}^n$, $\mathbf{e}_{k+1} \in \mathbb{R}^n$ being the standard bases.

We can solve~\eqref{opt:obj-complex} by computing the first principal component of the middle matrix of the right hand of~\eqref{opt:obj-complex}. Using the result from equation~\eqref{eq:varphi}, we have
    $\gamma =0.5 \text{atan2}(b_y,b_x),$ 
where 
 $   b_x = b_{x1} -b_{x2}.$ 
Since the subspace distance is the distance between $\mathbf{q}$ and $\mathbf{e}_k$, it is apparent that the subspace distance is $|\gamma|$. To maximize $|\gamma|$, we first determine the sign of $b_y$ or $b_x$. We have 
\begin{align}\nonumber
&\frac{b_x}{\|\mathbf{b}_k+\mathbf{e}_k\sigma_k\| + \| \mathbf{b}_{k+1} + \mathbf{e}_{k+1}\sigma_{k+1}\|} \\ \nonumber
&=
\|\mathbf{b}_k+\mathbf{e}_k\sigma_k\| -\| \mathbf{b}_{k+1} + \mathbf{e}_{k+1}\sigma_{k+1}\| \\\nonumber
&\ge \sigma_k-\|\mathbf{b}_k\| - \sigma_{k+1} - \|\mathbf{b}_{k+1}\| \\
&\ge \sigma_k - \sigma_{k+1} - \sqrt{2}\eta \label{ineq:bk} \\
&>0,\label{ineq:eta}
\end{align}
where inequality~\eqref{ineq:bk} is the result of the energy constraint that $\eta \ge \|\mathbf{B}\|_{\textrm{F}} = \sqrt{\|\mathbf{b}_k\|^2+\|\mathbf{b}_{k+1}\|^2} \ge \frac{1}{\sqrt{2}}(\|\mathbf{b}_k\|+\|\mathbf{b}_{k+1}\|)$, and inequality~\eqref{ineq:eta} is due to the assumption that $\eta < \frac{\sigma_k - \sigma_{k+1}}{\sqrt{2}}$. In summary, $b_x$ is positive. 
Using the property of atan2 function, when $b_x > 0$, maximizing $|\gamma|$ is equivalent to maximizing $|b_y/b_x|$. Thus, we can formulate our problem as
\begin{align}
    \max_{\mathbf{b}_k,\mathbf{b}_{k+1}}:\quad &|b_y/b_x| 
    \label{opt:no-rank-p-theta}\\\nonumber
    \text{s.t.}\quad & \|[\mathbf{b}_k, \mathbf{b}_{k+1}] \|_{\text{F}} \le \eta.
\end{align}
In the objective function,
\begin{align*}
    b_y=&2 \big(\mathbf{b}_1^\top \mathbf{b}_2 + (b_{k,k}+\sigma_k)b_{k+1,k}+ b_{k,k+1}(b_{k+1,k+1}+\sigma_{k+1}) \big), \\ \nonumber
    b_x=& \|\mathbf{b}_1\|^2 - \|\mathbf{b}_2\|^2 + (b_{k,k}+\sigma_k)^2+b_{k,k+1}^2-b_{k+1,k}^2\\  &-(b_{k+1,k+1}+\sigma_{k+1})^2,
\end{align*}
where $\mathbf{b}_1=[b_{k,1},\,b_{k,2},\,\cdots,\,b_{k,k-1},\,b_{k,k+2},\,\cdots,\,b_{k,n}]^\top$ and \\
$\mathbf{b}_2=[b_{k+1,1},\,b_{k+1,2},\,\cdots,\,b_{k+1,k-1},\,b_{k+1,k+2},\,\cdots,\,b_{k+1,n}]^\top$ which are the vectors obtained by deleting the $k$th and ($k+1$)th elements of $\mathbf{b}_k$ and $\mathbf{b}_{k+1}$ respectively. We can change the sign of $b_y/b_x$ by changing the signs of $\mathbf{b}_1, b_{k+1,k}$, and $b_{k,k+1}$. Since both of the values $b_y/b_x$ and $-b_y/b_x$ are obtainable, we can remove the absolute value operation. Thus, our objective can be further simplified to maximize $b_y/b_x$. To complete the proof of Theorem~\ref{theorem:no-rank-optimality}, we should further demonstrate that when the optimality of our objective function is obtained, $\mathbf{b}_1$ and $\mathbf{b}_2$ should be vectors with all their entries being zero. To prove that, we examine the objective function further 
\begin{align}\nonumber 
    b_y \le    
    & 2 \big(\|\mathbf{b}_1\| \|\mathbf{b}_2\| + (b_{k,k}+\sigma_k)b_{k+1,k}\\
    &+ b_{k,k+1}(b_{k+1,k+1}+\sigma_{k+1})\big) \label{ineq:no-rank-by1}  \\\nonumber
    \le
    & 2 \big((b_{k,k}+\sigma_k)b_{k+1,k}\\
    & + \sqrt{b_{k,k+1}^2+\|\mathbf{b}_1\|^2}(\sqrt{b_{k+1,k+1}^2+\|\mathbf{b}_2\|^2}+\sigma_{k+1}) \big), \label{ineq:no-rank-by2} \\\nonumber
    b_x \ge 
    &(b_{k,k}+\sigma_k)^2+b_{k,k+1}^2+ \|\mathbf{b}_1\|^2-b_{k+1,k}^2 \\
    &- (\sqrt{b_{k+1,k+1}^2+\|\mathbf{b}_2\|^2}+\sigma_{k+1})^2. \label{ineq:no-rank-bx}
\end{align}
Inequality~\eqref{ineq:no-rank-by1} implies that the optimal value is determined by the norms of $\mathbf{b}_1$ and $\mathbf{b}_2$ instead of their specific values. Inequality~\eqref{ineq:no-rank-by2} is true as
\begin{align*}
&\sqrt{b_{k,k+1}^2+\|\mathbf{b}_1\|^2}(\sqrt{b_{k+1,k+1}^2+\|\mathbf{b}_2\|^2}+\sigma_{k+1}) \\
&= \sqrt{b_{k,k+1}^2+\|\mathbf{b}_1\|^2}\sqrt{b_{k+1,k+1}^2+\|\mathbf{b}_2\|^2}\\ 
&\quad+\sigma_{k+1} \sqrt{b_{k,k+1}^2+\|\mathbf{b}_1\|^2} \\
&\ge b_{k,k+1}b_{k+1,k+1}+\|\mathbf{b}_1\|\|\mathbf{b}_2\| + \sigma_{k+1}b_{k,k+1}\\
&= \|\mathbf{b}_1\|\|\mathbf{b}_2\| + b_{k,k+1}(b_{k+1,k+1}+\sigma_{k+1}). 
\end{align*}
Inequality~\eqref{ineq:no-rank-bx} is due to  $-(\sqrt{b_{k+1,k+1}^2+\|\mathbf{b}_2\|^2}+\sigma_{k+1})^2 \le -\|\mathbf{b_2}\|^2 - (b_{k+1,k+1}+\sigma_{k+1})^2$. 
The equalities in~\eqref{ineq:no-rank-by2} and \eqref{ineq:no-rank-bx} hold when $\|\mathbf{b}_1\|=0$ and $\|\mathbf{b}_2\|=0$. This means that, for any feasible solution $(\mathbf{b}_1, \mathbf{b}_2, b_{k,k}, b_{k,k+1}, b_{k+1,k}, b_{k+1,k+1})$ in \eqref{opt:no-rank-p-theta}, there is another corresponding feasible solution
$(\mathbf{0},\mathbf{0},b_{k,k},\sqrt{b_{k,k+1}^2+\|\mathbf{b}_1\|^2}, b_{k+1,k}, \sqrt{b_{k+1,k+1}^2+\|\mathbf{b}_2\|^2})$, which has a larger objective value. In conclusion, $\mathbf{b}_1$ and $\mathbf{b}_2$ should be zero vectors when the optimality of \eqref{opt:no-rank-p-theta} is obtained. This completes our proof.



\bibliographystyle{./format/IEEEtran}
\bibliography{./format/IEEEabrv,mybib}

\begin{thebibliography}{10}
\providecommand{\url}[1]{#1}
\csname url@samestyle\endcsname
\providecommand{\newblock}{\relax}
\providecommand{\bibinfo}[2]{#2}
\providecommand{\BIBentrySTDinterwordspacing}{\spaceskip=0pt\relax}
\providecommand{\BIBentryALTinterwordstretchfactor}{4}
\providecommand{\BIBentryALTinterwordspacing}{\spaceskip=\fontdimen2\font plus
\BIBentryALTinterwordstretchfactor\fontdimen3\font minus
  \fontdimen4\font\relax}
\providecommand{\BIBforeignlanguage}[2]{{%
\expandafter\ifx\csname l@#1\endcsname\relax
\typeout{** WARNING: IEEEtran.bst: No hyphenation pattern has been}%
\typeout{** loaded for the language `#1'. Using the pattern for}%
\typeout{** the default language instead.}%
\else
\language=\csname l@#1\endcsname
\fi
#2}}
\providecommand{\BIBdecl}{\relax}
\BIBdecl

\bibitem{li2019adversarial}
F.~Li, L.~Lai, and S.~Cui, ``On the adversarial robustness of subspace
  learning,'' in \emph{Proc. IEEE International Conference on Acoustics, Speech
  and Signal Processing}, Brighton, UK, May 2019, pp. 2477--2481.

\bibitem{Li17compressvideo}
Y.~{Li}, W.~{Dai}, J.~{Zou}, H.~{Xiong}, and Y.~F. {Zheng}, ``Structured sparse
  representation with union of data-driven linear and multilinear subspaces
  model for compressive video sampling,'' \emph{IEEE Transactions on Signal
  Processing}, vol.~65, no.~19, pp. 5062--5077, Oct. 2017.

\bibitem{xin11subspacedoa}
J.~{Xin}, N.~{Zheng}, and A.~{Sano}, ``Subspace-based adaptive method for
  estimating direction-of-arrival with luenberger observer,'' \emph{IEEE
  Transactions on Signal Processing}, vol.~59, no.~1, pp. 145--159, Jan. 2011.

\bibitem{shen17onlinebigdata}
Y.~{Shen}, M.~{Mardani}, and G.~B. {Giannakis}, ``Online categorical subspace
  learning for sketching big data with misses,'' \emph{IEEE Transactions on
  Signal Processing}, vol.~65, no.~15, pp. 4004--4018, Aug. 2017.

\bibitem{guo2014online}
H.~Guo, C.~Qiu, and N.~Vaswani, ``An online algorithm for separating sparse and
  low-dimensional signal sequences from their sum,'' \emph{IEEE Transactions on
  Signal Processing}, vol.~62, no.~16, pp. 4284--4297, Aug. 2014.

\bibitem{otazo2015low}
R.~Otazo, E.~J. Cand{\`e}s, and D.~K. Sodickson, ``Low-rank plus sparse matrix
  decomposition for accelerated dynamic {MRI} with separation of background and
  dynamic components,'' \emph{Magnetic Resonance in Medicine}, vol.~73, no.~3,
  pp. 1125--1136, Apr. 2015.

\bibitem{koren2009matrix}
Y.~Koren, R.~Bell, and C.~Volinsky, ``Matrix factorization techniques for
  recommender systems,'' \emph{Computer}, vol.~42, no.~8, pp. 30--37, Aug.
  2009.

\bibitem{mardani2012dynamic}
M.~Mardani, G.~Mateos, and G.~B. Giannakis, ``Dynamic anomalography: Tracking
  network anomalies via sparsity and low rank,'' \emph{IEEE Journal of Selected
  Topics in Signal Processing}, vol.~7, no.~1, pp. 50--66, Feb. 2012.

\bibitem{guo2016video}
H.~Guo and N.~Vaswani, ``Video denoising via online sparse and low-rank matrix
  decomposition,'' in \emph{Proc. IEEE Statistical Signal Processing Workshop},
  Palma de Mallorca, Spain, Jun. 2016, pp. 1--5.

\bibitem{candes2011robust}
E.~J. Cand{\`e}s, X.~Li, Y.~Ma, and J.~Wright, ``Robust principal component
  analysis?'' \emph{Journal of the ACM}, vol.~58, no.~3, pp. 11:1--11:37, Jun.
  2011.

\bibitem{hsu2011robust}
D.~Hsu, S.~M. Kakade, and T.~Zhang, ``Robust matrix decomposition with sparse
  corruptions,'' \emph{IEEE Transactions on Information Theory}, vol.~57,
  no.~11, pp. 7221--7234, Jun. 2011.

\bibitem{qiu2014recursive}
C.~Qiu, N.~Vaswani, B.~Lois, and L.~Hogben, ``Recursive robust {PCA} or
  recursive sparse recovery in large but structured noise,'' \emph{IEEE
  Transactions on Information Theory}, vol.~60, no.~8, pp. 5007--5039, Jun.
  2014.

\bibitem{chen2011robust}
Y.~Chen, H.~Xu, C.~Caramanis, and S.~Sanghavi, ``Robust matrix completion and
  corrupted columns,'' in \emph{Proc. International Conference on Machine
  Learning}, Bellevue, Washington, Jun. 2011, pp. 873--880.

\bibitem{eykholt2017robust}
K.~Eykholt, I.~Evtimov, E.~Fernandes, B.~Li, A.~Rahmati, C.~Xiao, A.~Prakash,
  T.~Kohno, and D.~Song, ``Robust physical-world attacks on deep learning
  models,'' \emph{arXiv preprint arXiv:1707.08945}, Jul. 2017.

\bibitem{carlini2016hidden}
N.~Carlini, P.~Mishra, T.~Vaidya, Y.~Zhang, M.~Sherr, C.~Shields, D.~Wagner,
  and W.~Zhou, ``Hidden voice commands,'' in \emph{Proc. {USENIX} Security
  Symposium}, Austin, TX, Aug. 2016, pp. 513--530.

\bibitem{finlayson2019adversarial}
S.~G. Finlayson, J.~D. Bowers, J.~Ito, J.~L. Zittrain, A.~L. Beam, and I.~S.
  Kohane, ``Adversarial attacks on medical machine learning,'' \emph{Science},
  vol. 363, no. 6433, pp. 1287--1289, Mar. 2019.

\bibitem{chen2017targeted}
X.~Chen, C.~Liu, B.~Li, K.~Lu, and D.~Song, ``Targeted backdoor attacks on deep
  learning systems using data poisoning,'' \emph{arXiv preprint
  arXiv:1712.05526}, Dec. 2017.

\bibitem{golub2012matrix}
G.~H. Golub and C.~F. Van~Loan, \emph{Matrix computations}.\hskip 1em plus
  0.5em minus 0.4em\relax The Johns Hopkins University Press, 2013.

\bibitem{edelman1998geometry}
A.~Edelman, T.~A. Arias, and S.~T. Smith, ``The geometry of algorithms with
  orthogonality constraints,'' \emph{SIAM journal on Matrix Analysis and
  Applications}, vol.~20, no.~2, pp. 303--353, Apr. 1998.

\bibitem{weinstein2000almost}
A.~Weinstein, ``Almost invariant submanifolds for compact group actions,''
  \emph{Journal of the European Mathematical Society}, vol.~2, no.~1, pp.
  53--86, Mar. 2000.

\bibitem{georgiou1990optimal}
T.~T. Georgiou and M.~C. Smith, ``Optimal robustness in the gap metric,''
  \emph{IEEE Transactions on Automatic Control}, vol.~35, no.~6, pp. 673--686,
  Jun. 1990.

\bibitem{vinnicombe1993frequency}
G.~Vinnicombe, ``Frequency domain uncertainty and the graph topology,''
  \emph{IEEE Transactions on Automatic Control}, vol.~38, no.~9, pp.
  1371--1383, Sep. 1993.

\bibitem{qui1992feedback}
L.~Qui and E.~Davison, ``Feedback stability under simultaneous gap metric
  uncertainties in plant and controller,'' \emph{Systems \& Control Letters},
  vol.~18, no.~1, pp. 9--22, Jan. 1992.

\bibitem{jagielski2018manipulating}
M.~Jagielski, A.~Oprea, B.~Biggio, C.~Liu, C.~Nita-Rotaru, and B.~Li,
  ``Manipulating machine learning: Poisoning attacks and countermeasures for
  regression learning,'' in \emph{Proc. IEEE Symposium on Security and
  Privacy}, San Francisco, CA, May 2018, pp. 19--35.

\bibitem{bayraktar2019adversarial}
E.~Bayraktar and L.~Lai, ``On the adversarial robustness of multivariate robust
  estimation,'' \emph{arXiv preprint arXiv:1903.11220}, 2019.

\bibitem{pimentel2017adversarial}
D.~L. Pimentel-Alarcón, A.~Biswas, and C.~R. Solís-Lemus, ``Adversarial
  principal component analysis,'' in \emph{Proc. IEEE International Symposium
  on Information Theory}, Aachen, Germany, Jun. 2017, pp. 2363--2367.

\bibitem{huang2015role}
J.~Huang, Q.~Qiu, and R.~Calderbank, ``The role of principal angles in subspace
  classification,'' \emph{IEEE Transactions on Signal Processing}, vol.~64,
  no.~8, pp. 1933--1945, Apr. 2015.

\bibitem{he97detectiongapmetric}
C.~He and J.~M. Moura, ``Robust detection with the gap metric,'' \emph{IEEE
  Transactions on Signal Processing}, vol.~45, no.~6, pp. 1591--1604, Jun.
  1997.

\bibitem{absil2006largest}
P.~A. Absil, A.~Edelman, and P.~Koev, ``On the largest principal angle between
  random subspaces,'' \emph{Linear Algebra and its applications}, vol. 414,
  no.~1, pp. 288--294, Apr. 2006.

\bibitem{zimmermann2017closed}
R.~Zimmermann, ``A closed-form update for orthogonal matrix decompositions
  under arbitrary rank-one modifications,'' \emph{arXiv preprint
  arXiv:1711.08235}, Nov. 2017.

\bibitem{horn2012matrix}
R.~A. Horn and C.~R. Johnson, \emph{Matrix analysis}.\hskip 1em plus 0.5em
  minus 0.4em\relax Cambridge university press, 2012.

\bibitem{pertRank}
R.~C. Thompson, ``The behavior of eigenvalues and singular values under
  perturbations of restricted rank,'' \emph{Linear Algebra and its
  Applications}, vol.~13, no. 1-2, pp. 69--78, 1976.

\bibitem{jackson2005user}
J.~E. Jackson, \emph{A user's guide to principal components}.\hskip 1em plus
  0.5em minus 0.4em\relax John Wiley \& Sons, 2005, vol. 587.

\bibitem{kalivas1997two}
J.~H. Kalivas, ``Two data sets of near infrared spectra,'' \emph{Chemometrics
  and Intelligent Laboratory Systems}, vol.~37, no.~2, pp. 255--259, Jun. 1997.

\end{thebibliography}
\end{document}